\documentclass[sigconf,nonacm]{acmart}

\usepackage[boxruled, lined, linesnumbered, algosection, noend]{algorithm2e}
\usepackage{xspace}
\usepackage{subcaption}

\usepackage{balance}
\usepackage{orcidlink}

\usepackage{booktabs}

\usepackage{cleveref}

\newtheorem{theorem}{Theorem}
\newtheorem{lemma}{Lemma}
\newtheorem{corollary}{Corollary}

\SetAlgoSkip{bigskip}
\SetKwComment{Comment}{/* }{ */}
\SetKwInOut{Input}{Input}
\SetKwInOut{Output}{Output}
\SetKwRepeat{Repeat}{repeat}{until}
\SetKwBlock{RepeatOnly}{repeat}{}

\Crefname{algocf}{Algorithm}{Algorithms}
\crefname{enumi}{}{}

\newcommand{\textcite}[1]{\citeauthor{#1} \cite{#1}}

\let\cref\Cref

\newcommand{\EIS}{\texttt{EIS}\xspace}
\newcommand{\EISm}{\texttt{EISm}\xspace}
\newcommand{\NIS}{\texttt{NIS}\xspace}
\newcommand{\EISlong}{Edge Induced Subgraph\xspace}
\newcommand{\NISlong}{Node Induced Subgraph\xspace}
\newcommand{\multipassbaseline}{\texttt{3ES}\xspace}
\let\threeES\multipassbaseline
\newcommand{\CCN}{\texttt{ColoredChibaNishizeki}\xspace}
\newcommand{\Abacus}{\textsc{Abacus}\xspace}
\newcommand{\Fleet}{\textsc{Fleet3}\xspace}

\newcommand{\instancenamestyle}[1]{\texttt{#1}\xspace}
\newcommand{\movielens}{\instancenamestyle{movielens}}
\newcommand{\dblp}{\instancenamestyle{dblp}}
\newcommand{\reuters}{\instancenamestyle{reuters}}
\newcommand{\livejournalgroups}{\instancenamestyle{livejournal-groups}}
\newcommand{\trackers}{\instancenamestyle{trackers}}
\newcommand{\orkut}{\instancenamestyle{orkut}}
\newcommand{\pubmed}{\instancenamestyle{pubmed}}
\newcommand{\flickr}{\instancenamestyle{flickr}}
\newcommand{\socLiveJournal}{\instancenamestyle{soc-\allowbreak{}LiveJournal}}
\newcommand{\friendster}{\instancenamestyle{friendster}}

\newcommand{\biohuman}{\instancenamestyle{bio-human-gene}}
\newcommand{\hollywood}{\instancenamestyle{hollywood}}
\newcommand{\bnhuman}{\instancenamestyle{bn-human-Jung}}

\newcommand{\IP}{\mathbb{P}}
\newcommand{\IE}{\mathbb{E}}
\newcommand{\IV}{\operatorname{Var}}

\newcommand{\CT}{\mathcal{T}}

\renewcommand{\epsilon}{\varepsilon}
\newcommand{\eps}{\varepsilon}
\newcommand{\bigO}{\widetilde{O}}
\newcommand{\tildeO}{\widetilde{O}}

\newcommand{\poly}{\mathrm{poly}}

\newcommand{\squares}{\CT}

\newcommand{\Var}[1]{\IV\left[ #1 \right]}
\newcommand{\EXP}[1]{\IE\left[ #1 \right]}
\newcommand{\Prob}[1]{\IP\left[ #1 \right]}

\newcommand{\pluseq}{\mathrel{+}=}

\newcommand{\storedEdges}{\emph{\#stored edges}\xspace}

\begin{document}
\title{Four-Cycle Counting in Low-Degeneracy Graph Streams}
\titlenote{Authors ordered alphabetically.}

\author{Sebastian Lüderssen}
\orcid{0009-0002-5698-1250}
\affiliation{
    \institution{TU Wien}
    \city{Vienna}
    \country{Austria}
}
\email{sebastian.luederssen@tuwien.ac.at}

\author{Stefan Neumann}
\orcid{0000-0002-3981-1500}
\affiliation{
    \institution{TU Wien}
    \city{Vienna}
    \country{Austria}
}
\email{stefan.neumann@tuwien.ac.at}

\author{Pan Peng}
\orcid{0000-0003-2700-5699}
\affiliation{
  \institution{University of Science and Technology of China}
  \city{Hefei}
  \country{China}
}
\email{ppeng@ustc.edu.cn}

\renewcommand{\shortauthors}{Sebastian Lüderssen, Stefan Neumann, \& Pan Peng}

\date{}

\begin{abstract}
We study the problem of $(1+\varepsilon)$-approximating the number of four-cycles in graphs given as arbitrary order edge streams. We propose two new algorithms based on sampling induced subgraphs. Our first contribution is a two-pass algorithm that uses $\widetilde{O}(\kappa m / \sqrt{T})$ space, where $m$ is the number of edges, $T$ is the number of four-cycles, and $\kappa$ is the graph’s degeneracy. This algorithm improves upon existing theoretical bounds and is provably optimal for constant-degeneracy graphs, matching the known $\Omega(m/\sqrt{T})$ lower bound up to lower-order factors. Our second contribution is a one-pass algorithm that remains accurate when four-cycles are not highly concentrated around individual nodes, edges, or wedges; this structural property is common in sparse social and collaboration networks. We evaluate both algorithms on a variety of real-world graph streams.  The two-pass algorithm consistently outperforms state-of-the-art methods, using substantially less space to achieve a desired accuracy. The one-pass algorithm is competitive when four-cycles are evenly distributed, matching our theoretical analysis. Unlike several recent works, our algorithms perform well even on non-bipartite graphs such as social networks.
\end{abstract}

\maketitle

\section{Introduction}
Graphs are one of the most fundamental abstract data types in computer science,
as they allow modeling data from diverse fields, ranging from social networks
and transportation networks to user--item purchases. 
As data volumes continue to grow, modern applications increasingly involve large-scale graphs with billions of edges.

Subgraph counting is a well-established tool for obtaining insights into these
massive graphs and it is used in many applications, such as analyzing the human brain \cite{sporns2004motifs}, fraud detection in insurance and cryptocurrency
networks~\cite{oskarsdottirSocialNetworkAnalytics2022,lalUnderstandingMoneyTrails2021} and for user clustering in
social networks~\cite{charbey2019stars,rotabi2017detectingStrong}, among many other applications. Depending on the
concrete setting, popular subgraph patterns include triangles, four-cycles,
or small cliques.
Here, we focus on counting \emph{four-cycles}, i.e., cycles with four vertices.
In bipartite graphs, four-cycles are typically referred to as \emph{butterflies} and take the role of triangles in unipartite graphs.
They give rise to the butterfly clustering coefficient, a standard measure of cohesiveness in bipartite networks~\cite{lind2005cycles}, and are 
relevant in online recommendation systems~\cite{abuodaMA19LinkPrediction}.
Four-cycles also play an important role in unipartite networks; for instance, it was shown that in social networks, individuals who are contained in many four-cycles are also  wealthier~\cite{shridhar2025network}. 

However, the size of modern networks has led to many computational challenges,
such as the datasets not fitting into main memory. One of the paradigms to deal
with these challenges is the \emph{streaming} setting, where algorithms make a small number of sequential passes over the graph's edges while using as little memory as possible to analyze its structure. 
While one-pass streaming algorithms are ideal in some settings, multi-pass algorithms --- where the input is read multiple times in the same order --- are often practical. In external-memory systems, sequential multi-pass scans are cheaper than large RAM usage or random disk access. Many real-world workflows (e.g., ETL pipelines, Hadoop/Spark jobs) already involve multiple scans, making multi-pass algorithms a natural fit.

This has led to considerable interest in subgraph counting in the streaming setting, resulting in the development of numerous algorithms (e.g.,~\cite{lim2015mascot,stefani2016triest,bordinoMiningLargeNetworks2008,jayaramOptimalAlgorithmTriangle2021,fichtenbergerApproximatelyCountingSubgraphs2022}). Four-cycle counting has also received significant attention~\cite{sanei-mehriFLEETButterflyEstimation2019,papadiasCountingButterfliesFully2024,liApproximatelyCountingButterflies2022,kallaugherMPV19,vorotnikovaImproved3passAlgorithm2020}, with most existing work focusing on one-pass streaming algorithms, often specialized for bipartite graphs.

In this paper, we study one of the simplest and most versatile algorithms for four-cycle
counting in the streaming setting: Use one or two stream passes to collect a
small induced subgraph of the original graph, then count the number of four-cycles on the induced subgraph and scale it up to obtain a global estimate. This algorithm has the advantage that during the stream passes, it only spends constant time for each edge it processes, and we can rely on efficient
non-streaming algorithms on our small sample subsequently.
Another advantage of using this simple way of collecting an induced subgraph is that, besides allowing to count four-cycles, it maintains many topological properties of the graph very well~\cite{ahmed2013NetworkSampling}, allowing further insights into the graph at no additional memory or computational overhead.

\textbf{Our results.} We propose two novel space-efficient algorithms for counting four-cycles in arbitrary order edge streams. Both follow the aforementioned paradigm: they sample an induced subgraph during their stream pass(es) and then estimate the number of four-cycles afterwards. Our algorithms work well on bipartite graphs, as well as on unipartite graphs.

Our main result is a two-pass streaming algorithm, \EIS, that returns a
$(1+\varepsilon)$-approximation of the number of four-cycles in the graph and
uses space\footnote{We write $\tildeO(\cdot)$ to hide $\poly(\log m, 1/\eps)$ factors.}
 $\bigO(\kappa m/\sqrt{T})$, where $m$ is the number of edges in the
graph, $T$~is the number of four-cycles that shall be estimated, and $\kappa$ is
the degeneracy of the graph (see \Cref{sec:algorithms} for a formal definition). This algorithm is highly space-efficient since it
is known that real-world networks typically have very low degeneracy due to
their sparsity~\cite{eppstein2013listing,shin2018Patterns}. Moreover, the class of constant-degeneracy graphs encompasses many natural graph families, such as planar graphs, minor-closed families, and preferential attachment graphs, among others.
Our algorithm extends a $\bigO(m^{3/2}/T^{3/4})$-space algorithm for
distinguishing graphs with $0$ and $T$ four-cycles by
\textcite{mcgregorTriangleFourCycle2020} and we give a more detailed discussion after \Cref{thm:EIS}.

For sparse graphs, this improves upon the best
currently known theoretical result, which requires a space usage of
$\bigO(m/T^{1/3})$ by \citet{vorotnikovaImproved3passAlgorithm2020}, i.e., we provide an improvement as long as $\kappa = O(T^{1/6})$.  For $\kappa=\bigO(1)$ and up to lower order terms, the space usage of our algorithm is optimal, as it matches a theoretical space
lower bound of $\Omega(m/\sqrt{T})$ by \citet{mcgregorTriangleFourCycle2020}, which holds for graphs with
$\kappa=O(1)$ and for any algorithm that performs $O(1)$~passes over the stream. 
Additionally, the two passes are indeed necessary to obtain this result since
known lower bounds imply that single-pass algorithms require space
$\Omega(m)$ in general~\cite{kallaugherMPV19} (even in the stronger adjacency list streaming model).

Interestingly, the sampling procedure used by \EIS is very simple: During the first stream pass, sample a set of edges uniformly at random, and during the second stream pass, store the subgraph induced by these edges.
This sampling approach was also studied by 
\citet{ahmed2013NetworkSampling} who showed that it performs well for obtaining insights into the network topology (such as tails of degree distributions, core sizes and eigenvalues), though they did not consider subgraph counting, and their analysis is very different from ours. Thus, our work shows that this sampling approach provides a single snapshot that is useful for estimating multiple quantities, including four-cycle counts. Furthermore, since our algorithm's space usage depends on the degeneracy, the algorithm is highly efficient for sparse networks.
This is similar to the work of \citet{beraHowDegeneracyHelps2020} who also exploited the degeneracy of the network to obtain more space-efficient triangle counting algorithms (see also \cite{fichtenbergerApproximatelyCountingSubgraphs2022} for the more general case of $r$-clique counting); however, their algorithm is much more complicated than \EIS and uses different techniques.

Further, we present a single-pass streaming algorithm, \NIS, that returns a
$(1+\varepsilon)$-approximation of the number of four-cycles in the graph and
uses space $\bigO(m/\sqrt{T})$ if some technical conditions hold, which intuitively state that the four-cycles should be nicely distributed across the graph (see
\Cref{sec:theory-nis} for the details).

We implement our algorithms 
 and provide practical improvements over
their theoretical versions. In our experiments, we show that (averaged over 7~real-world datasets) \EIS requires
more than 10~times less memory than state-of-the-art competitors to obtain a mean error of less than 5\%, while still being substantially faster in terms of running time
--- even though it performs two stream passes.  We also present a version
of \EIS, which has reduced variance and which returns robust estimates using
almost 3~times less space than the competitors. We also show that \NIS obtains good results on
sparse unipartite datasets.
Our code is available online~\cite{code}.

\section{Related Work}\label{sec:relatedwork}
Subgraph counting in data streams is highly popular, both in practical communities, as well as in more theoretical ones. We review these two strands of research for four-cycle counting, and we also discuss related results for general subgraph counting.

Regarding practical streaming algorithms for four-cycle counting,
most recent works focus on bipartite graph streams.
\citet{sanei-mehriFLEETButterflyEstimation2019} presented
\Fleet, a sampling-based algorithm for estimating
four-cycle counts in bipartite insertion-only streams. It is based on edge sampling and uses only one stream pass. \Fleet relies on sampling three out of four edges of a four-cycle to detect it; this is more restrictive than our method \EIS, which only has to sample two out of four edges of a four-cycle to detect it.
\citet{liApproximatelyCountingButterflies2022} proposed another algorithm for the insertion-only setting in bipartite networks, which samples edges as well and computes estimates for the number of shared neighbors per node pair, storing only the largest values by employing an AMS-sketch~\cite{alon1999space}.
\citet{papadiasCountingButterfliesFully2024} presented the \textsc{Abacus} algorithm for fully-dynamic graph streams, which
achieves performance comparable to or better than that
of~\cite{liApproximatelyCountingButterflies2022} in the insertion-only setting.

In more general settings,
\citet{sunFABLEApproximateButterfly2024} and \citet{mengCountingButterfliesStreaming2024}
study multi-graph streams with duplicate edges, and
\citet{sheshboloukiSGrappButterflyApproximation2022} consider four-cycle counting in the sliding window streaming model.
\citet{bordinoMiningLargeNetworks2008} proposed a three-pass algorithm for counting four-cycles and other subgraphs on three or four nodes.

Most of the theoretical work about four-cycle counting in graph streams considers the following variant of the problem:
Given a graph $G$ as an arbitrary order insertion-only edge stream and a constant-factor approximation of the
number of four-cycles $T$ in $G$, output a $(1+\epsilon)$-approximation of $T$ with high probability. 

In the one-pass arbitrary order model, there are instances where $\Omega(m)$
space is required to decide whether a graph has $0$ or at least $T$ four-cycles for any $T=O(m^{1/3})$~\cite{kallaugherMPV19}, thus essentially ruling out efficient streaming algorithms for general $T$ from a theoretical standpoint (even for the stronger adjacency list streaming model).

For settings in which we can do $O(1)$~passes over the stream, 
\citet{beraTighterSpaceBounds2017} gave a simple two-pass algorithm for counting an arbitrary subgraph $H$ in space $\tildeO(m^{\beta(H)}/T)$, where $\beta(H)$ denotes the integral edge cover number of $H$.
This was improved by \citet{assadiSimpleSublinearTimeAlgorithm2018} and \citet{fichtenbergerApproximatelyCountingSubgraphs2022} to obtain a three-pass $\tildeO(m^{\rho(H)}/T)$ space algorithm for general subgraph patterns~$H$, where $\rho(H)$ is the fractional edge cover number of $H$. For four-cycles these algorithms achieve a space usage of $\tildeO(m^2/T)$.
\citet{mcgregorTriangleFourCycle2020} gave a new upper bound using $\tildeO(m/T^{1/4})$ space and three passes, which was improved to space $\tildeO(m/T^{1/3})$ by \citet{vorotnikovaImproved3passAlgorithm2020}. 
The best known constant-pass lower bound is $\Omega(m/\sqrt{T})$ by \citet{mcgregorTriangleFourCycle2020}, and our results match this bound for graphs with degeneracy $\kappa=\bigO(1)$ (and we note that the lower bound also holds for $\kappa=O(1)$). 

Recently, in \citet{luederssen2026near}, we presented a near-optimal three-pass streaming algorithm for counting four-cycles using $\tildeO(m/\sqrt{T})$ space for general graphs. While this algorithm has strong theoretical guarantees, it is too complicated to be implemented in practice. In contrast, the present paper focuses on practical algorithms that only require one or two passes and are tailored to low-degeneracy graphs; the two works were conducted independently and address different application scenarios.

\section{Algorithms}
\label{sec:algorithms}

In this section, we present our practical algorithms for counting four-cycles in
insertion-only edge streams of arbitrary order.  Unlike previous algorithms,
our algorithms are based on collecting induced subgraphs, which are either induced by
a random sample of edges or nodes, respectively.

We study undirected, unweighted graphs
$G=(V,E)$ with $n$~nodes and $m$~edges. We assume that we
obtain the edges~$E$ as an arbitrary-order stream. Given a set of
nodes~$V_S\subseteq V$ we will frequently consider the \emph{induced subgraph}
$G[V_S] = (V_S, E_S)$, which is the subgraph with node set~$V_S$ and all edges
$E_S = \{ (u,v) \in E \colon u, v \in V_S\}$ with both endpoints in $V_S$. We
write $d(v)$ to denote the degree of a vertex~$v$.
We will also talk about \emph{wedges}, which are paths of length~2.

The \emph{degeneracy $\kappa$} of a graph $G$ is defined as the maximum of the minimum degrees over all subgraphs of $G$. Up to a constant factor, it is the same as the \emph{arboricity} of the graph, which is the minimum number of forests into which its edges can be partitioned. Degeneracy is typically considered as a sparsity measure, since real-world graphs tend to have very low degeneracy (see also \Cref{tab:datasets}).

We let $T$ denote the number of four-cycles in the graph.

\subsection{Edge Induced Subgraph \EIS}
\label{sec:EIS}

We start by presenting our two-pass algorithm, \EISlong (\EIS), and we present
its pseudocode in \Cref{algo:eis}. Here, we present a version which samples
edges with a given probability~$p$, and we refer to \Cref{sec:implementation} for a
version of \EIS that stores a fixed-size number of edges.  Furthermore, in
\Cref{sec:new_m_alpha_sqrtT} we prove that with small changes, \EIS returns a
$(1+\varepsilon)$-approximate answer using space
$\bigO(m\kappa/(\varepsilon^{3}\sqrt{T}))$ and time $\tildeO(m+\kappa^2 m/\sqrt{T})$.

\begin{algorithm}[t]
    \caption{\EISlong \EIS}
    \label{algo:eis}
    \DontPrintSemicolon
    \textbf{Input:} Graph $G$, sampling probability $p$\;
    \SetKwProg{Fn}{Procedure}{:}{}
    \Fn{\EIS($G$, $p$)}{
    \textbf{Pass 1:} Sample an edge set  $S$ by sampling each edge independently with probability $p$\;
	Let $V_S$ be the nodes incident to at least one edge in $S$\;
    \textbf{Pass 2:} Collect all edges in $G[V_S]$\;
	$X = $ \CCN($G[V_S]$, $S$)\;
    \tcc{$X$ is the number of four-cycles in $G[V_S]$ with two opposite-side edges in $S$}
    \Return{$X/(2p^2)$}
	}
	\bigskip
    \Fn{\texttt{ColoredChibaNishizeki}($G[V_S]$, $S$)}{
    \label{algo:BicoloredChiba}
		$C=0$\;
		Assign color~0 to all edges in~$S$ and color~$1$ to all edges in $G[V_S]$\;
		Sort the nodes by degree s.t. $d(v_1)\geq d(v_2)\geq \dots \geq d(v_n)$\;
		\For{$v\in V_S$ in non-increasing degree order}{
			\ForEach{$u\in N_0(v)$}{
				\ForEach{$w\in N_1(u)$ with $v\neq w$}{
					$U_{01}[v,w] \pluseq 1$\;
				}
                \ForEach{$w\in N_0(u)$ with $v\neq w$}{
					$U_{00}[v,w] \pluseq 1$\;
				}
			}
			\ForEach{$u\in N_1(v)$}{
				\ForEach{$w\in N_0(u)$ with $v\neq w$}{
					$U_{10}[v,w] \pluseq 1$\;
				}
			}
			\ForEach{$w$ with $U_{01}[v,w]>0$}{
				$C \pluseq U_{01}[v,w]\cdot U_{10}[v,w]-U_{00}[v,w]$
			}
			Remove $v$ and all of its incident edges from $G[V_S]$\;
		}
        \Return{$C$}\;
	}
\end{algorithm}

\emph{High-level description.} \EIS works as follows. In the first stream pass, \EIS samples a set of
edges $S$ uniformly at random, where each edge is sampled independently with
probability~$p$.  We let $V_S = \{u\in V \colon \exists (u,v) \in S\}$
denote the set of nodes incident upon an edge in $S$.  In the second stream pass, \EIS
collects the induced subgraph~$G[V_S]$, i.e., the subgraph induced by the edges
in~$S$.

After the second pass, \EIS counts the number of four-cycles
in~$G[V_S]$ (we describe how this is done efficiently below) and estimates the global four-cycle count $T$ by scaling up the four-cycle count in $G[V_S]$.
Crucially, when counting the number of four-cycles in $G[V_S]$, \EIS only counts those
four-cycles in $G[V_S]$, which are induced by two opposite-side edges from our
edge sample $S$. Specifically, we only count a four-cycle
$v_1$--$v_2$--$v_3$--$v_4$--$v_1$ if $(v_1,v_2)\in S$ and
$(v_3,v_4)\in S$ or if $(v_2,v_3)\in S$ and $(v_4,v_1)\in S$. Let $X$ denote the number of such four-cycles; \EIS then outputs $X / (2p^2)$.

The intuition why this algorithm should be efficient is as follows (we make this
more formal in our theoretical analysis below):
(1)~Each four-cycle~$v_1$--$v_2$--$v_3$--$v_4$--$v_1$ is counted in~$G[V_S]$
with probability $2p^2$, since we sample the opposite-side edges $(v_1,v_2)$ and
$(v_3,v_4)$ with probability $p^2$, and we also have probability $p^2$ to sample
the other pair of opposite-side edges $(v_2,v_3)$ and $(v_4,v_1)$. Thus, in
expectation we sample $2 T p^2$ four-cycles and if we want to find $\bigO(1)$ of
them, it suffices to set $p=\frac{1}{\sqrt{T}}$. In contrast, the methods by
\citet{mcgregorTriangleFourCycle2020} and \citet{vorotnikovaImproved3passAlgorithm2020} all require setting $p$ to be at least $\frac{1}{T^{1/3}}$ to find the same number of
four-cycles, which is significantly larger.
(2)~If the original graph has bounded degeneracy, then the induced subgraph~$G[V_S]$
will be sparse, implying that the space to store the induced edges in $G[V_S]$
will be very small. In our theoretical analysis, we will formalize this by arguing
about the degeneracy of $G$ and $G[V_S]$.

\emph{Counting four-cycles in~$G[V_S]$ with opposite-side edges from~$S$.}
Next, we explain how to count the number of four-cycles in~$G[V_S]$ with
opposite-side edges from~$S$ efficiently. We propose an extension of the popular
algorithm by \citet{chibaArboricitySubgraphListing1985} and adapt it to
bi-edge-colored graphs; its pseudocode is stated in \Cref{algo:eis}
in the procedure~\CCN. 

We first introduce some notation. Each edge is assigned one or two colors: if an edge is in~$S$, it is assigned color~$0$; 
all edges in~$G[V_S]$ are assigned color~$1$. As edges $e\in S$ are also
contained in $G[V_S]$, such edges receive both colors $0$ and $1$. For a
vertex~$v$, $N_0(v)$ (resp.\ $N_1(v)$) denotes the set of neighbors connected to $v$ via edges in~$S$ (resp.\ $G[V_S]$). Note that $N_0(v)$ and $N_1(v)$ are not necessarily disjoint, and that edges connecting~$v$
to nodes in $N_0(v)$ (resp.\ $N_1(v)$) have color~$0$ (resp.\ $1$). Next, we maintain counters $U_{01}[v,w]$ (resp.\ $U_{10}[v,w]$), which store the number of wedges $v$--$u$--$w$ that use a $0$-edge (resp. $1$-edge) $v$--$u$ and a
$1$-edge (resp. $0$-edge) $u$--$w$. Then, a four-cycle
$v$--$u_1$--$w$--$u_2$--$v$, which is induced by opposite-side edges $(v,u_1)\in S$ and
$(w,u_2)\in S$, corresponds to a pair of alternatingly colored wedges between $v$ and $w$, which leads to $U_{10}[v,w]$ and $U_{01}[v,w]$ being increased by one.
Thus, the counters $U_{01}[v,w]$ and $U_{10}[v,w]$ allow us to efficiently identify and count only those four-cycles induced by opposite-side edges.

Given this notation, \CCN works as follows. It initializes a counter $C=0$ and orders the vertices in
non-increasing order of their degrees and iterates over them.  During an
iteration for a vertex~$v$ of highest degree, it proceeds as follows.  It
iterates over all of $v$'s neighbors $u\in N_0(v)$. For each such neighbor~$u$,
it iterates over all of $u$'s neighbors~$w\in N_1(u)$ with $w\neq v$ and increments $U_{01}[v,w]$. The counter $U_{10}[v,w]$ is created similarly. Additionally, a counter $U_{00}[v,w]$ is updated, which counts wedges using two $0$-edges.
After this has finished, we observe that the number of four-cycles with alternating colors containing $v$ and $w$ as non-adjacent vertices is given by
\begin{equation*}
    U_{01}[v,w]\cdot U_{10}[v,w]-U_{00}[v,w],
\end{equation*}
since any pair of wedges counting towards $U_{01}[v,w]$ and $U_{10}[v,w]$ defines a four-cycle induced by opposite edges in $S$, except those using the same center vertex. Thus, we need to subtract the number of wedges for which both edges are assigned both colors. This is tracked by the counter $U_{00}[v,w]$.
After updating the counter~$C$, we remove~$v$ from the graph and proceed to the next iteration until all nodes in $V_S$ have been processed. Finally, the algorithm returns~$C$.

Similar to \cite{chibaArboricitySubgraphListing1985}, one can show that \CCN
counts the four-cycles with opposite-side edges in~$S$ in time $O(m_S \kappa_S)$ and space $O(m_S)$, where $m_S=\tildeO(\kappa m/\sqrt{T})$ is the number of edges in $G[V_S]$ and $\kappa_S$ is the degeneracy of~$G[V_S]$. As
$\kappa_S \leq \kappa$, this algorithm is highly efficient and runs in
near-linear time in practice.
Thus, the overall running time of \EIS is $\tildeO(m+\kappa^2 m/\sqrt{T})$.

\subsection{Node Induced Subgraph \NIS}
\label{sec:NIS}
Next, we present our one-pass algorithm \NIS and provide its pseudocode in
\Cref{sec:appendix_algorithms}. \NIS is also based on sampling an induced subgraph, but, unlike
\EIS, it samples nodes instead of edges.

More precisely, \NIS works as follows. We let $p\in[0,1]$ be a sampling
probability and we let $h \colon V\rightarrow[0,1]$ be an $8$-wise-independent
universal hash function. Now we can obtain a random node set $V_h$ of $G$ by letting
$V_h=\{v\in V \colon h(v)< p\}$, i.e., each node is included in $V_h$ with
probability $p$. We note that this node set is not sampled explicitly, but
membership of any node~$v$ in $V_h$ can be checked in constant time by
evaluating the hash function and checking if $h(v) < p$.  Thus, during the
stream pass, we collect all edges in the induced graph $G[V_h]=(V_h,E_h)$ by
evaluating the hash function on both endpoints of an incoming edge $e=(v,w)$ and
storing $e$ if $v$ and $w$ are in $V_h$, i.e., if $h(v)<p$ and $h(w)<p$.
At the end, we count the number of four-cycles in $G[V_h]$, which we denote by $X$, and return $X/p^4$.

Note that \NIS counts a four-cycle if and only if it sampled all of its
vertices. 
As this happens with probability $p^4$, the returned value $X/p^4$ is an unbiased estimator
for $T$.

\subsection{Implementation Details}
\label{sec:implementation}
Recall that, so far, both \NIS and \EIS expect a given sampling probability $p$
as input. However, this is not very practical since we will see below that
setting $p$ requires approximate knowledge of~$m$ and~$T$.
Instead, it is more practical to
work with a hard constraint,~$k$, on the number of edges that our algorithms may
store. Thus, we adapt our algorithms to this setting and present variants of our
algorithms, which take $k$ as input and subsequently never store more than
$k$~induced edges from the graph stream. We also present \EISm, a version of \EIS with reduced variance.

\emph{\EIS with fixed-size budget.}
To adapt \EIS to never store more than $k$~edges in the induced subgraph, we
need to sample a random edge set~$S$ such that $G[V_S]$ contains at most
$k$~edges.  As the final size of our edge sample~$S$ or of
$G[V_S]$ is unknown beforehand, our algorithm initially samples ``too many'' edges in~$S$
and then removes edges from~$S$ such that $G[V_S]$ never stores
more than $k$~edges.

The size-constrained version of \EIS starts by sampling a set $S$ of $k$ edges
in the first pass using reservoir sampling \cite{vitter85random}.  In the second pass, it 
collects edges induced by $V_S$ as long as $G[V_S]$ contains at most
$k$~edges.  Once we would exceed our space budget by adding a $(k+1)$-st
edge, we prune~$S$ and $G[V_S]$ as follows: We select a
random edge $e=(v,w)\in S$ and remove it from $S$.  Now, if $v$ has no more
incident edges in~$S$, we delete $v$ from $V_S$ and remove all edges
from~$G[V_S]$ incident to $v$. We do the same for $w$. We
repeat this procedure until $G[V_S]$ stores at most $k$~edges.

We emphasize that this fixed-budget variant is an implementation-oriented
modification of the Bernoulli-sampling version (i.e., \Cref{algo:eis_oracle}) analyzed theoretically. Since
pruning is triggered by the current size of the induced subgraph $G[V_S]$, the
final sample $S$ should not be treated as an independent Bernoulli sample, nor
as a uniformly random edge subset conditional only on its final size. 
Additionally, at the end of the second pass, we have stored all edges in
$G[V_S]$ and $G[V_S]$ never has more than $k$~edges. 

Finally, note that since we did not sample with a fixed probability~$p$, we
have to redefine how we scale up our estimate of the number of four-cycles.
Concretely, after counting $X$ four-cycles in $G[V_S]$ with opposite-side edges
in~$S$, we set $\tilde{p}=|S|/m$ and we return the estimate $X/(2\tilde{p}^2)$.
Observe that we can set $\tilde{p}=|S|/m$ because after the stream pass we
know~$m$ and that $\tilde{p}$ is picked such that it matches the fact that our
final edge sample contains $|S|$~edges.

\emph{\EISm: Reducing the variance of \EIS.} 
In our preliminary experiments, we found that \EIS provides excellent results w.r.t.\ the mean (and also the median) of its results, even when using very little space (e.g., just storing 1,000 or 2,000 edges). However, when using more space its variance does not drop as quickly as for other methods (see also \Cref{sec:experiments-variance} for a further discussion).

Thus, we also consider a version of \EIS, called \EISm (which is short for \EIS-mean), which takes as input a parameter~$k$ for how many edges it might store in total. It further considers another parameter~$s\in\mathbb{N}$.
Now, instead of computing one induced subgraph, \EISm runs $s$ copies of \EIS, each of which can store up to $k/s$~edges and then we take the average over the $s$~estimates returned by the $s$~different copies of \EIS.
We run experiments with this variant for $s=32$.

\emph{\NIS with fixed-size budget.}
We can adapt \NIS to a fixed-size budget in a similar way, by dynamically
deleting vertices from the sample.  However, here we cannot use reservoir sampling, and
instead we vary the sampling probability $p$ in the algorithm.
We describe the details in \cref{sec:appendix_algorithms}.

\section{Theoretical Guarantees}
\label{sec:theory}

In this section, we state and formally prove our theoretical results.
For readability, we present all the proofs in \cref{sec:appendix_proofs}.

\subsection{Guarantees for \EIS}
\label{sec:new_m_alpha_sqrtT}

We present our guarantees for a slightly modified version of \EIS. 
\begin{theorem}
\label{thm:EIS}
	There is a two-pass streaming algorithm that, given an arbitrary order edge
	stream of a graph $G$, uses space $\bigO(\kappa m/(\epsilon^3\sqrt{T}))$ and
	returns a $(1+\epsilon)$-approximation of the number of four-cycles in~$G$
	with probability\footnote{We remark that by using the standard median trick, we can boost the constant success probability of \Cref{thm:EIS} and \Cref{thm:NIS} to $1-\delta$, with an additional cost of $O(\log(1/\delta))$ in the space complexity, for any $\delta>0$.} $\frac{2}{3}$ .
\end{theorem}

To prove this, our main challenge will be to bound the
variance of our estimator. We do this by restricting the set of four-cycles
in the sample that count towards our estimate; in particular, we have to
avoid counting four-cycles containing \emph{heavy} edges, i.e., edges that are
included in many four-cycles, as such edges affect the variance substantially.
This is done by introducing a heaviness oracle, which returns for any queried
edge whether it is heavy, and which allows us to exclude any heavy edges
from our sample.

Our algorithm extends a four-cycle detection algorithm from \cite{mcgregorTriangleFourCycle2020}.  
Our analysis adds three parts: (1)~We adapt the approach in \cite{mcgregorTriangleFourCycle2020} to obtain a
$(1+\epsilon)$-approximate four-cycle \emph{counting} algorithm, given a
heaviness oracle. (2)~We present a heaviness oracle, which can be built using two stream
passes using only a constant factor of extra space and which answers queries correctly with high probability.
(3)~We analyze the space complexity in terms of the degeneracy. 

\emph{Algorithm.}
The theoretical version of \EIS samples two edge sets $S$ and $R$ in the
first pass and lets $V_S$ and $V_R$ denote the nodes incident to these sets. Here, the sampling probability is set to $p=320\log(n)/(\epsilon^3\sqrt{T})$, where we assume that we know a lower bound on~$T$ (this is a standard assumption in this area, see, e.g., the discussions in \cite[Sec.~1]{braverman2013hard} and \cite[Sec.~1.2]{mcgregorBetterAlgorithmsCounting2016}). Then the algorithm collects the subgraph $G[V_S\cup V_R]$ in the second pass. Edges in $R$ are used to build an oracle to classify edges in $G[V_S]$ as heavy (corresponding to high variance) or light. 
We then use the independent sample~$S$ to estimate~$T$ by counting four-cycles induced by pairs of edges in $S$ that contain at most one heavy edge.
Concretely, we let $A_0$ ($A_1$) denote the number of such four-cycles without (with one) oracle-heavy edge. Now the final estimate is given by $A_0/(2p^2) + A_1/p^2$.
We state the pseudocode in \Cref{sec:pseudocode-EIS-theory}.

We say two edges $e$ and $f$ \emph{induce a four-cycle} if they are vertex-disjoint and there is a four-cycle containing both $e$ and $f$. Note that two edges can induce up to two four-cycles.

\emph{Space analysis.}
Both $R$ and $S$ are of size $\Theta(mp)=\widetilde{\Theta}(m/(\epsilon^3\sqrt{T}))$ with high probability.
Thus ${G[V_S\cup V_R]}$ contains $\bigO(\kappa m/(\epsilon^3\sqrt{T}))$ edges
since any $n'$-node subgraph of a graph with degeneracy~$\kappa$ contains
$O(\kappa n')$~edges (see, e.g., \cite{chibaArboricitySubgraphListing1985}).
Counting the four-cycles in the graph sample can be done without any additional space usage.

\emph{Heaviness.}
For $e\in E$, let $t(e)$ denote the number of four-cycles containing edge $e$, and
set $\eta=\frac{82}{\epsilon}$. We define an edge to be \textit{heavy} if $t(e)\geq \eta\sqrt{T}$ and \textit{light} otherwise. We use the following structural lemma from \textcite{mcgregorTriangleFourCycle2020}.

\begin{lemma}[\cite{mcgregorTriangleFourCycle2020}]
\label{lemma:heavyedges}
    At most $\epsilon T$ four-cycles contain two or more heavy edges.
\end{lemma}

\emph{Oracle.}
Let $q(e)$ be the number of four-cycles induced by $e\in G[V_S]$ and an edge in $R$.
Observe that we can compute $q(e)$ in our algorithm since we collected all edges
in $G[V_S \cup V_R]$ and thus, since $e\in G[V_S]$ and the opposite-side edge in $R$ are vertex-disjoint, any four-cycle counted in $q(e)$ only uses edges from $G[V_S \cup V_R]$.

Now, for edges~$e\in G[V_S]$, the oracle is defined as 
\begin{equation}
\label{eq:oracle}
    \operatorname{oracle}(e)=\begin{cases}
        \textsf{light}, & \text{if } q(e) \leq 2\eta p\sqrt{T},\\
        \textsf{heavy}, & \text{if } q(e) >2\eta p\sqrt{T}.
        \end{cases}
\end{equation}
Note that the algorithm only calls the oracle for edges in $G[V_S]$.

We call edges that are classified as light (resp.\ heavy) by the oracle as
\emph{oracle-light} (resp.\ \emph{oracle-heavy}). 
\Cref{lemma:oracle_accuracy} shows that we do not need to distinguish between light and oracle-light (resp. heavy and oracle-heavy) edges.
\begin{lemma}
\label{lemma:oracle_accuracy}
    With probability $1-\frac{1}{n}$, for all $e\in G[V_S]$ it holds that,
    $\operatorname{oracle}(e)=\textsf{light}$ implies $t(e)\leq 4\eta \sqrt{T}$ and 
    $\operatorname{oracle}(e)=\textsf{heavy}$ implies $t(e)\geq \eta \sqrt{T}$.
\end{lemma}

\emph{Counting accuracy.} 
Let $\mathcal{T}_i$ denote the set of four-cycles 
containing exactly $i$ oracle-heavy edges and let $T_i=|\mathcal{T}_i|$ for $i\in \{0,1\}$. 
By \Cref{lemma:oracle_accuracy}, oracle-heavy edges are heavy with high probability, and thus by \Cref{lemma:heavyedges},
\begin{equation}
    (1-\epsilon)T\leq T_0+T_1\leq T. \label{eq:c4_T0T1}
\end{equation}

Now let $\hat{T_0}=A_0/2p^2$ and  $\hat{T_1}=A_1/p^2$  denote the algorithm's estimates for $T_0$ and $T_1$. 
To analyze $A_0$ and $A_1$, we consider vertex-disjoint pairs of light edges that the algorithm might pick.
Let $D_i$ denote the set of oracle-light edge pairs inducing at least one four-cycle of $\mathcal{T}_i$. 
For $i\in\{0,1\}$ we define for each edge-pair $q\in D_i$ a random variable $X_q$ which counts the four-cycles induced by $q$ with exactly $i$ oracle-heavy edges, if $q\subseteq S$, and is 0, otherwise. Note that $X_q\leq2$.

We now have 
\begin{align*}
    \EXP{\hat{T_0}}&= \EXP{\frac{A_0}{2p^2}} = \frac{1}{2p^2}\sum_{q\in D_0} \IE[X_q]=\frac{1}{2p^2}\sum_{C\in \mathcal{T}_0} 2p^2=T_0
\end{align*} 
since every four-cycle in $\mathcal{T}_0$ is induced by two pairs of oracle-light edges. Similarly, $\EXP{\hat{T_1}}= T_1$, 
since every four-cycle in $\mathcal{T}_1$ is induced by only one pair of oracle-light edges. 

Now we show that $\hat{T_i}$ is accurate with probability~$\frac{9}{10}$.
\begin{lemma}
\label{lemma:eis_variance}
With probability $\frac{9}{10}$, $|\hat{T_i}-T_i|\leq \epsilon T$ for $i\in\{0,1\}$.
\end{lemma}
\allowdisplaybreaks
We can now prove \Cref{thm:EIS}.
\begin{proof}[Proof of \Cref{thm:EIS}]
    
\Cref{lemma:eis_variance} together with \eqref{eq:c4_T0T1} implies
    $(1-3\epsilon)T\leq \hat{T_0}+\hat{T_1}\leq (1+2\epsilon)T$.
By setting $\epsilon'=\epsilon/6$ and running the algorithm for $\epsilon'$ we have $1/(1+\epsilon)\leq 1-3\epsilon'$ and thus get a $(1+\epsilon)$ approximation on the number of four-cycles $T$.
\end{proof}

\emph{Guarantees for \EIS without oracle.}
As mentioned in \Cref{sec:implementation}, we do not use any heaviness oracle in our experiments, as it does not improve the performance of the algorithm.
If we neglect the oracle, i.e., by setting $\operatorname{oracle}(e)=\textsf{light}$ for all edges $e$, we obtain a slightly weaker result, which --- together with our experimental results --- explains the high quality of \EIS solutions in practice.
Let $\Delta_E = \max_{e\in E} t(e)$ be the maximum edge heaviness, i.e., the maximum number of four-cycles sharing a single edge. 
If $\Delta_E\leq\sqrt{T}$, there are no heavy edges in $G$ and the algorithm correctly assumes that all edges are light. This immediately yields the following result, which also improves upon the dependency on $\frac{1}{\varepsilon}$ because we do not use the oracle here. Furthermore, as we report in \Cref{tab:datasets}, the assumption of $\Delta_E \leq \sqrt{T}$ is reasonable on real-world datasets.
\begin{corollary}
\label{cor:EIS}
    \Cref{algo:eis} (without a heaviness oracle) is a two-pass streaming algorithm with space complexity $\bigO(\kappa m /(\epsilon^2\sqrt{T}))$ and returns a $(1+\epsilon)$-approximation of the number of four-cycles in~$G$ with probability $\frac{2}{3}$ if $\Delta_E\leq\sqrt{T}$.
\end{corollary}

\subsection{Guarantees for \NIS}
\label{sec:theory-nis}
We now provide theoretical guarantees for \NIS (see \Cref{algo:nis}). In the theorem, we let 
$\Delta_V$ denote the maximum number of four-cycles sharing a single vertex,
$\Delta_E$ denote the maximum number of four-cycles sharing a single edge, and
$\Delta_W$ denote the maximum number of four-cycles sharing a single wedge.

\begin{theorem}
\label{thm:NIS}
\NIS is a 1-pass streaming algorithm using space $O(m/(\epsilon^4\sqrt{T}))$, and it returns a $(1+\epsilon)$-approximation on the number of four-cycles with probability $\frac{5}{8}$, if the graph fulfills $\Delta_V\leq T^{3/4}$, $\Delta_E\leq \sqrt{T}$ and $\Delta_W\leq T^{1/4}$.
\end{theorem}

We validate the theorem's assumptions on real-world datasets in \Cref{tab:datasets} and show that they hold on unipartite social-network datasets, whereas bipartite graphs typically have larger values for $\Delta_V$ and~$\Delta_W$.

\section{Experiments}
\label{sec:experiments}

We empirically evaluate our algorithms \EIS,  \EISm and \NIS on real-world datasets and compare them to previous works.
Our implementation is available online~\cite{code}.
Our goal is to answer the following questions:
\begin{enumerate}
    \item How much space is necessary for a good average error?
    \item How large is the variance of the algorithms' estimates?
    \item How computationally efficient are the algorithms?
\end{enumerate}

\emph{Datasets.} We evaluate our algorithms on 13 real-world networks of varying size, taken from the KONECT network library \cite{kunegisKONECTKoblenzNetwork2013} and the Network Repository \cite{nr}. 
For all datasets, we consider an undirected version and remove loops and duplicate edges. We use the provided edge list for each instance as the order for the edge stream.

\Cref{tab:datasets} presents the datasets and core properties, including their average degree $\bar{d}$, degeneracy $\kappa$ and the heaviness thresholds used in the theoretical analysis of our algorithms (stated as a ratio for better readability; see also \Cref{cor:EIS} and \Cref{thm:NIS}). The lower the given values of $\Delta_V/T^{3/4}$, $\Delta_E/T^{2/4}$, and $\Delta_W/T^{1/4}$, the lower the variance of our algorithms \NIS and \EIS. While all three of these ratios need to be small in order for \NIS to be accurate, for \EIS the variance depends solely on $\Delta_E$ as shown in \Cref{cor:EIS}. The datasets are chosen from various domains and reflect a wide range of structural properties.

\begin{table*}[t]
  \centering
  \small
  \caption{Statistics of our datasets. For each graph, we report its number of nodes~$n$ and edges~$m$. We also report the average degree~$\bar{d}$, the degeneracy~$\kappa$ and the exact number of four-cycles~$T$. Further, we report the maximum number of four-cycles per node~$\Delta_V$, per edge~$\Delta_E$, and wedge~$\Delta_W$, normalized by the heaviness thresholds from \Cref{thm:NIS}. Finally, we report $m/\sqrt{T}$, which serves as a lower bound on how much space is theoretically necessary for $O(1)$-pass algorithms, and the type of the graph.}
  \label{tab:datasets}
  \resizebox{\linewidth}{!}{%
  \begin{tabular}{lrrrrrrrrrrl}
    \toprule
    Instance            & $n$      & $m$        & $\bar{d}$ & $\kappa$ & $T$      & $\Delta_V/T^{3/4}$ & $\Delta_E/T^{1/2}$ & $\Delta_W/T^{1/4}$ & $m/\sqrt{T}$ & Type         \\
    \midrule
    \movielens          & 0.08M    & 10.00M    & 248.40 & 476 & 1.20T    &  16.54 &   6.29 &   24.20 &    9.14 & bipartite     \\
    \dblp                & 7.58M    & 12.28M    &   3.24 &  14 & 31.67M   &   0.25 &   0.13 &    5.68 & 2182.33 & bipartite     \\
    \reuters             & 1.07M    & 60.57M    & 113.74 & 192 & 7.49T    & 102.14 &   6.21 &   94.36 &   22.13 & bipartite     \\
    \livejournalgroups  & 10.69M   &112.31M    &  21.01 & 108 & 3.30T    & 283.44 &   5.52 &  338.99 &   61.85 & bipartite     \\
    \trackers            & 40.42M   &140.61M    &   6.96 & 437 &20.07T    &1083.61 &  10.44 & 1163.24 &   31.39 & bipartite     \\
    \orkut               & 11.51M   &327.04M    &  56.81 & 466 &22.13T    &  49.86 &   4.38 &   60.36 &   69.52 & bipartite     \\
    \pubmed              &  8.34M   &483.45M    & 115.85 & 108 &40.82T    & 244.24 &   2.85 &  251.55 &   75.65 & bipartite     \\
    \midrule
    \biohuman           &0.02M  & 12.32M        & 563 & 2047 & 16.13T   &  3.81  & 8.59  & 3.62 & 3.06  & unipartite \\
    \flickr             & 2.30M    & 22.84M    &  19.83 & 600 & 0.71T    & 4.56 &   5.86 &  7.75 &   27.06 & unipartite    \\
    \socLiveJournal     & 4.85M    & 42.85M    &  17.69 & 372 & 51.52B   & 1.47 &   2.24 &  9.62 &  188.79 & unipartite    \\
    \hollywood          &1.1M   & 56.37M        & 512 & 2208 & 4.96T    & 2.07  & 6.45   & 4.51 & 22.28    & unipartite \\
    \bnhuman            &0.78M  & 267.84M       & 683 & 1200 & 44.01T   & 2.90  & 8.9    & 7.96 & 40.37 & unipartite\\
    \friendster         & 68.34M   &  1.81B    & 53.02 & 304 & 470.2B    & 0.22 &   0.39 &  3.62 & 2642.37 & unipartite    \\
    \bottomrule
  \end{tabular}
  }
\end{table*}
We observe that across all datasets, 
$\Delta_E$ is small. 
This justifies that for \EIS the assumptions of \Cref{cor:EIS} are satisfied and we do not have to use the heaviness oracle from the theoretical analysis in practice.
Further, the unipartite social networks feature small values of $\Delta_V$ and $\Delta_W$, making them well-suited for both \NIS and \EIS. In contrast, the bipartite networks typically have large values of $\Delta_V$ and $\Delta_W$; here, 
\trackers is a particularly extreme case where half of all four-cycles share a single node. \dblp is an outlier due to its extreme sparsity ($\kappa=14$) and few four-cycles ($T$ is very small). 
We provide a thorough description of the datasets and their properties in \cref{sec:datasets}.

\emph{Algorithms.}
We implemented our algorithms \EIS, \EISm and \NIS using C\texttt{++} and make the source code publicly available \cite{code}.
We use the fixed-sample-size variants of our algorithms (see \Cref{sec:implementation}).
We note that we also implemented the theoretical version of \EIS from \Cref{sec:new_m_alpha_sqrtT}, but it did not yield better results in practice (which is also explained by the low $\Delta_E$-values reported in \Cref{tab:datasets} and \Cref{cor:EIS}), and therefore we omit it here.

We compare against the state-of-the-art algorithms \textsc{Fleet} by \citet{sanei-mehriFLEETButterflyEstimation2019} and \textsc{Abacus} by \citet{papadiasCountingButterfliesFully2024}.
Out of several algorithms in the \textsc{Fleet} suite, we run their best performing algorithm \Fleet using their C\texttt{++} implementation. We set the subsampling parameter of \Fleet to $\gamma=0.75$, as recommended by the authors. 
For \textsc{Abacus}, we run their publicly available Java implementation.
Both of these algorithms are one-pass algorithms. They sample edges from the input graph~$G$ during the stream pass and work with a given limit~$k$ on the size of that sample, i.e., the number of edges stored from $G$ never exceeds $k$.
We do not compare against the algorithm by \citet{liApproximatelyCountingButterflies2022} since in the experiments of \cite{papadiasCountingButterfliesFully2024} \textsc{Abacus} performs comparable or better than that of \cite{liApproximatelyCountingButterflies2022} and since their code is not available.

As \EIS uses two passes over the stream, while \NIS, \textsc{Fleet} and \textsc{Abacus} use only one stream pass, we additionally compare to a simple 2-pass estimator mentioned in the state-of-the-art theoretical work~\cite{vorotnikovaImproved3passAlgorithm2020}, which we denote by \multipassbaseline: In the first pass, we sample $k$ edges using reservoir sampling to obtain a uniformly random edge sample. 
In the second pass we do not store any further edges, but only count the number of four-cycles each edge completes with three edges from the sample. 
The total number of observed four-cycles is divided by $4p^3$ to obtain an unbiased estimator where $p=k/m$.
We note that \multipassbaseline is based on the same idea as the one-pass \Fleet algorithm but does not rely on subsampling. Interestingly, we will see that, unlike \Fleet, \multipassbaseline also performs well on unipartite graphs, indicating that the second stream pass is indeed helpful.

\emph{Experimental setup.}
To measure the space usage of the algorithms, we report how many edges they may store; this setup was also used by the previous works~\cite{sanei-mehriFLEETButterflyEstimation2019,papadiasCountingButterfliesFully2024,liApproximatelyCountingButterflies2022}. We will report this as the (number of) \storedEdges $k$.

We run all algorithms on 8 different sample sizes, ranging from~500 to~256k \storedEdges, increasing by factors of 2.
This setting is more ambitious than that of previous works, which used sample sizes ranging from 75k to 600k \storedEdges \cite{papadiasCountingButterfliesFully2024,sanei-mehriFLEETButterflyEstimation2019}. \EISm is run for $s=32$, i.e., it computes 32~smaller samples and outputs the average; we provide a sensitivity analysis of the parameter $s$ in \Cref{sec:appendix_experiments_sensitivity-s}.

All algorithms are run 100~times per dataset and sample size. We report averages over these 100~runs and further discuss their
variance below. 
For an algorithm's estimate~$\hat{T}$, 
we report the absolute relative error~$\frac{|\hat{T}-T|}{T}$, as well as
the signed relative error $\frac{\hat{T}-T}{T}$.
Note that a signed relative error of~$-1$
indicates that the algorithm did not find a single four-cycle and thus estimated
$\hat{T}=0$.

\subsection{Average Error Analysis}

\begin{table*}[t!]
\centering
\caption{Minimum \storedEdges required for each algorithm to achieve 10\% and 5\% relative error, respectively,  averaged over 100 runs. Averaged over all bipartite instances, \EIS requires 5.9~times less space to achieve 10\% and 11.2~times less space to achieve 5\%~error. A dashed line indicates the error bound is not met even at a sample size of 256k. }
\label{tab:error_table}
\resizebox{\linewidth}{!}{%
\begin{tabular}{lrrrrrrrrrrrr}
\toprule
Instance & \multicolumn{6}{c}{$<10\%$ Error} & \multicolumn{6}{c}{$<5\%$ Error} \\
 \cmidrule(lr){2-7} \cmidrule(lr){8-13}
 & \Fleet & \Abacus & \threeES & \NIS & \EIS & \EISm &  \Fleet & \Abacus & \threeES & \NIS & \EIS & \EISm  \\
\midrule
\movielens        & 2.0k       & 1.0k       & \textbf{0.5k}       & 8.0k       & \textbf{0.5k}       & 2.0k       & 2.0k       & 1.0k       & 1.0k       & 32.0k      & \textbf{0.5k}       & 4.0k       \\
\dblp                 & 32.0k      & 16.0k      & 64.0k      & 8.0k       & \textbf{2.0k}       & 16.0k      & 32.0k      & 16.0k      & 64.0k      & 8.0k       & \textbf{2.0k}       & 16.0k      \\
\reuters              & 1.0k       & 2.0k       & 8.0k       & 64.0k      & \textbf{0.5k}       & \textbf{0.5k}       & 8.0k       & 8.0k       & 16.0k      & 128.0k     & \textbf{0.5k}       & 2.0k       \\
\livejournalgroups   & 8.0k       & 8.0k       & 1.0k       & 256.0k     & \textbf{0.5k}       & \textbf{0.5k}       & 16.0k      & 8.0k       & 4.0k       & 256.0k     & 1.0k       & \textbf{0.5k}       \\
\trackers             & 4.0k       & 32.0k      & 8.0k       & –          & \textbf{0.5k}       & 1.0k       & 4.0k       & 32.0k      & 8.0k       & –          & \textbf{0.5k}       & 2.0k       \\
\orkut                & 2.0k       & 16.0k      & 8.0k       & 16.0k      & \textbf{0.5k}       & \textbf{0.5k}       & 4.0k       & 32.0k      & 32.0k      & 64.0k      & \textbf{0.5k}       & \textbf{0.5k}       \\
\pubmed               & 16.0k      & 32.0k      & 8.0k       & –          & \textbf{0.5k}       & 2.0k       & 16.0k      & 32.0k      & 16.0k      & –          & \textbf{0.5k}       & 4.0k       \\
\midrule
\biohuman           & –          & –          & \textbf{0.5k}       & 2.0k       & \textbf{0.5k}       & 4.0k       & –          & –          & \textbf{0.5k}       & 8.0k       & 1.0k       & 32.0k      \\
\flickr               & –          & –          & \textbf{0.5k}       & \textbf{0.5k}       & \textbf{0.5k}       & \textbf{0.5k}       & –          & –          & \textbf{0.5k}       & 2.0k       & \textbf{0.5k}       & 4.0k       \\
\socLiveJournal      & –          & –          & 8.0k       & 8.0k       & \textbf{1.0k}       & \textbf{1.0k}       & –          & –          & 8.0k       & 8.0k       & \textbf{1.0k}       & 4.0k       \\
\hollywood          & –          & –          & 2.0k       & 1.0k       & \textbf{0.5k}       & \textbf{0.5k}       & –          & –          & 2.0k       & 4.0k       & \textbf{0.5k}       & \textbf{0.5k}       \\
\bnhuman            & –          & –          & 2.0k       & \textbf{0.5k}       & \textbf{0.5k}       & \textbf{0.5k}       & –          & –          & 2.0k       & 1.0k       & \textbf{0.5k}       & \textbf{0.5k}       \\
\friendster           & –          & –          & 128.0k     & \textbf{2.0k}       & 4.0k       & 16.0k          & –          & –          & 256.0k     & \textbf{2.0k}       & 16.0k      & 32.0k           \\
\bottomrule
\end{tabular}
}
\end{table*}

First, we study the minimum \storedEdges for which the algorithms achieve a mean relative error at most 10\% and 5\%, respectively.
We present the results in \Cref{tab:error_table}.

Our two-pass algorithm \EIS consistently achieves an accurate mean estimate using very small graph samples.
For example, only considering 500 out of 480 million edges on \pubmed yields estimates with less than 5\% error on average.
In fact, \EIS approximates the four-cycle count of all instances up to 5\% error using just 1k edges, except for the two graphs with the least four-cycle density, \dblp and \friendster, which both have $m/\sqrt{T}$-ratios of over 2,000. This matches the space lower bound of $\Omega(m/\sqrt{T})$ (see \Cref{tab:datasets}), indicating that estimating $T$ on these datasets requires more space.
While our variance-reduced variant \EISm requires slightly larger \storedEdges than \EIS for accurate mean estimates, both of our two-pass algorithms outperform the one-pass previous works \Fleet and \Abacus substantially in terms of mean accuracy on most instances. Further, they are also clearly more efficient than the two-pass baseline \multipassbaseline, showing that the increased accuracy is not due to the additional stream pass.
The only exception is the small and dense \movielens dataset, which also has the largest four-cycle density among all our instances, where \multipassbaseline, \Fleet and \Abacus almost match the performance of \EIS.
As expected, \NIS performs worse on the bipartite datasets, since they have large $\Delta_V$ and $\Delta_W$ values.

As  \Fleet and \Abacus are only designed for bipartite graph streams, they do not correctly count four-cycles on the unipartite networks \flickr, \socLiveJournal and \friendster. 
Our one-pass \NIS algorithm generates accurate estimates on average for the unipartite instances, beating the edge-sampling-based approaches on the \friendster network.

We note that the performance of \NIS matches our theoretical analysis very well: 
It performs well on the unipartite networks, as well as on the bipartite \dblp dataset, where $\Delta_V$, $\Delta_E$ and $\Delta_W$ are small.
Indeed, \NIS has 5\%~error when just storing 8k~edges on \dblp. However, it fails to produce good estimates on \trackers and \pubmed, even when storing 256k~edges, due to the large $\Delta_V$- and $\Delta_W$-values of the datasets.
These results highlight that the performance of \NIS as described in \Cref{thm:NIS} aligns very well with the dataset properties reported in \Cref{tab:datasets}: On \trackers, \pubmed and \livejournalgroups the normalized maximum node and wedge heaviness values, $\Delta_V$ and $\Delta_W$, which indicate how well four-cycles are distributed over the graph, are orders of magnitude higher than on \dblp and the unipartite instances, and this directly transfers to higher variance and worse estimates of the \NIS estimator.

\begin{table*}[t]
\centering
\caption{Minimum \storedEdges required to achieve 10\% error on 50\% and 90\% of the individual runs for each algorithm. Averaged over all bipartite instances, \EISm beats baselines by a factor of 6.1 for a 50\% success rate and 2.8 for a 90\% success rate.}
\label{tab:variance_table}
\resizebox{\linewidth}{!}{%
\begin{tabular}{lrrrrrrrrrrrr}
\toprule
Instance & \multicolumn{6}{c}{$<10\%$ Error on $>50\%$ of runs} & \multicolumn{6}{c}{$<10\%$ Error on $>90\%$ of runs} \\
 \cmidrule(lr){2-7} \cmidrule(lr){8-13}
 & \Fleet & \Abacus & \threeES & \NIS & \EIS & \EISm &  \Fleet & \Abacus & \threeES & \NIS & \EIS & \EISm  \\
\midrule
\movielens        & \textbf{4.0k}       & \textbf{4.0k}       & \textbf{4.0k}       & 64.0k      & 8.0k       & \textbf{4.0k}       & \textbf{8.0k}       & \textbf{8.0k}       & 16.0k      & –          & 256.0k     & 32.0k      \\
\dblp                 & 128.0k     & 128.0k     & 128.0k     & 64.0k      & \textbf{32.0k}      & 64.0k      & 256.0k     & 256.0k     & 256.0k     & 256.0k     & \textbf{64.0k}      & 256.0k     \\
\reuters              & 32.0k      & 16.0k      & 32.0k      & –          & 8.0k       & \textbf{2.0k}       & 64.0k      & \textbf{32.0k}      & 64.0k      & –          & 128.0k     & \textbf{32.0k}      \\
\livejournalgroups   & 64.0k      & 64.0k      & 128.0k     & –          & 32.0k      & \textbf{16.0k}      & 256.0k     & 256.0k     & 256.0k     & –          & –          & \textbf{64.0k}      \\
\trackers             & 128.0k     & 128.0k     & 256.0k     & –          & 16.0k      & \textbf{8.0k}       & –          & 256.0k     & –          & –          & 128.0k     & \textbf{64.0k}      \\
\orkut                & 64.0k      & 64.0k      & 64.0k      & –          & \textbf{16.0k}      & \textbf{16.0k}      & \textbf{128.0k}     & \textbf{128.0k}     & \textbf{128.0k}     & –          & –          & \textbf{128.0k}     \\
\pubmed               & 128.0k     & 64.0k      & 128.0k     & –          & \textbf{8.0k}       & \textbf{8.0k}       & 256.0k     & 256.0k     & 256.0k     & –          & 128.0k     & \textbf{32.0k}      \\
\midrule
\biohuman           & –          & –          & \textbf{2.0k}       & 32.0k      & 16.0k      & 4.0k       & –          & –          & \textbf{8.0k}       & –          & 256.0k     & 32.0k      \\
\flickr               & –          & –          & 8.0k       & 256.0k     & 64.0k      & \textbf{4.0k}       & –          & –          & \textbf{32.0k}      & –          & –          & 64.0k      \\
\socLiveJournal      & –          & –          & 64.0k      & 256.0k     & 32.0k      & \textbf{16.0k}      & –          & –          & 128.0k     & –          & –          & \textbf{64.0k}      \\
\hollywood          & –          & –          & 16.0k      & –          & 128.0k     & \textbf{4.0k}       & –          & –          & \textbf{32.0k}      & –          & –          & 64.0k      \\
\bnhuman            & –          & –          & 32.0k      & 16.0k      & \textbf{2.0k}       & \textbf{2.0k}       & –          & –          & 64.0k      & –          & 32.0k      & \textbf{16.0k}      \\
\friendster           & –          & –          & –          & \textbf{64.0k}      & \textbf{64.0k}      & 128.0k          & –          & –          & –          & \textbf{256.0k}     & \textbf{256.0k}     & –          \\
\bottomrule
\end{tabular}
}
\end{table*}

\subsection{Variance Analysis}
\label{sec:experiments-variance}
Second, we study the variance of the algorithms.
In \Cref{tab:variance_table} we present the minimum \storedEdges to achieve at most 10\% error in at least 50\% and at least 90\% of the individual runs.

In both the 50\% and 90\% individual accuracy regimes, our variance-reduced variant \EISm performs best across most instances. Exceptions are \movielens and \orkut, where \Fleet and \Abacus perform similarly, as well as the high degeneracy dataset \biohuman.
We observe that \EIS still beats \Fleet and \Abacus in the setting where 50\% of the estimates need to be accurate. However, the variance of \EIS remains too high to produce high-quality estimates on 90\% of runs on many instances (we give a justification for this in \cref{sec:appendix_experiments_inducedFraction}).

\NIS is able to obtain accurate predictions in at least 50\% of the runs on the unipartite datasets, but not in 90\% of the runs when storing at most 256k~edges.

A final important observation is that on the sparse \dblp network, the regular \EIS clearly outperforms the variance-reduced \EISm. We provide an explanation for this phenomenon in \cref{sec:appendix_experiments_inducedFraction}.

\subsection{Running Time Analysis}
Third, in \Cref{fig:runtime} we compare the running times of all algorithms by considering the time it took to output an estimate using 32k \storedEdges on all instances, and for different sample sizes on the \reuters network.
To avoid bias from how efficiently the algorithms read the file from disk, we loaded the graphs into main memory before running the algorithms.

\begin{figure}[t]
\newcommand{\smallfigwidth}{0.49\linewidth} 
\newcommand{\smallfigsep}{1em}               

    \makebox[\linewidth][c]{%
    \hspace{2em}
    \begin{subfigure}[b]{\linewidth}
      \centering
      \includegraphics[width=0.7\linewidth]{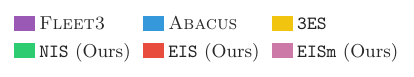}
    \end{subfigure}
  }

  \makebox[\smallfigwidth][c]{%
    \begin{subfigure}[b]{\smallfigwidth}
      \captionsetup{skip=-1pt}
      \centering
      \includegraphics[width=\linewidth]{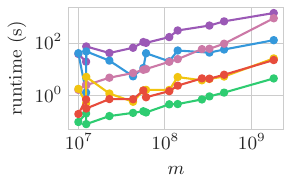}
      \caption{\small Runtime for 32k \storedEdges \phantom{aaaaaaaaaaaaaaaa}}
      \label{fig:runtime_m}
    \end{subfigure}
  }
  \hfill
  \makebox[\smallfigwidth][c]{%
    \begin{subfigure}[b]{\smallfigwidth}
      \captionsetup{skip=-1pt}
      \centering
      \includegraphics[width=\linewidth]{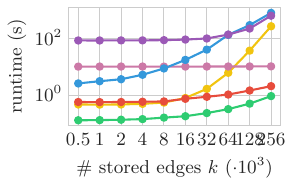}
      \caption{\small Runtime for varying \storedEdges on \reuters}
      \label{fig:runtime_k}
    \end{subfigure}
  }

  \caption{Comparison of the algorithms' running times on all datasets (\Cref{fig:runtime_m}) and for varying \storedEdges $k$ on \reuters (\Cref{fig:runtime_k}). \NIS is consistently the fastest method, followed by \EIS, which is still much faster than the baselines. \EISm's running time is between \Abacus and \Fleet, but scales better as the \storedEdges $k$ grows.}
  \label{fig:runtime}
\end{figure}

In \Cref{fig:runtime_m}, we find that the running times of our algorithms \NIS, \EIS and \EISm scale linearly with the number of edges in the instance. This is expected as these algorithms only require $O(1)$~time per edge in the stream and a single call to a four-cycle counting algorithm on the sample of a fixed size.
The \EIS algorithm takes twice as long as \NIS for one estimate, as it performs two passes over the stream instead of one.
The \EISm algorithm is approximately 32~times slower than \EIS, since we set $s=32$ and thus we need to process each incoming stream edge for each of the 32~subgraphs we are sampling.
In \cref{sec:appendix_experiments_sensitivity-s} we give a thorough sensitivity analysis for the parameter $s$ and find that choosing $s=32$ provides a reasonable trade-off between runtime and variance improvements.

The algorithms \multipassbaseline, \Fleet and \Abacus also scale roughly linearly in running time for fixed \storedEdges, but exhibit large differences in running times across instances with a similar number of edges.
In fact, these algorithms require up to $O(d^2)$ time upon arrival of an edge in the stream, where $d$ denotes the maximum node degree in the sample. 
Thus, for a fixed \storedEdges, the density of the network determines the running time, explaining higher execution times for the 10 million edge \movielens network in comparison to the next larger 12 million edge \dblp network.

When varying the \storedEdges in \Cref{fig:runtime_k}, the difference in runtime scaling between previous works and our algorithms is even more noticeable. 
On \reuters, \multipassbaseline, \Fleet and \Abacus show quadratic increases in runtime, whereas our algorithms do not consume any extra time per edge, and only have a small runtime overhead for computing the four-cycle count in the graph sample after the stream passes.

\NIS and \EIS compute estimates significantly faster than previous works, even when using two stream passes instead of one (in the case of \EIS).
Our \EISm variant requires more runtime than \EIS, scaling linearly in the number $s$ of samples to compute for the mean.
Although \EISm takes longer than \Abacus for small sample sizes, its superior scaling with respect to the \storedEdges and its higher accuracy make it highly practical, particularly when considering even larger networks with tens of billions of edges, where more space may be necessary.

\subsection{Summary and Further Experiments}
Summarizing our experiments, we arrive at the following recommendations for picking a suitable algorithm.
If two stream passes are feasible, we recommend \EIS for fast and reliable estimates in expectation, if moderate variance is acceptable. If very low variance is required, \EISm is preferable.
If only a single stream pass is feasible, we recommend \NIS for very sparse unipartite datasets. Otherwise, \Abacus is preferred for its higher accuracy, despite its longer runtime.

Further, we found that, across the evaluated datasets, the observed algorithm performance was closely tied to the node and wedge heaviness: when $\Delta_V$ and $\Delta_W$ are small (e.g., the unipartite social networks \flickr, \socLiveJournal, \friendster), \NIS produces accurate mean estimates, including strong performance on \friendster; likewise, on the bipartite \dblp instance (where $\Delta_V$ and $\Delta_W$ are also small), \NIS already reaches about 5\% error with only $8\text{k}$ stored edges. In contrast, on instances with very large $\Delta_V$ and $\Delta_W$ (most notably \trackers but also \pubmed and \livejournalgroups, whose normalized $\Delta_V$ and $\Delta_W$ are orders of magnitude larger than on \dblp and the unipartite graphs), \NIS exhibits substantially higher variance and fails to obtain good estimates even at $256\text{k}$ stored edges. Over the same range of graphs, \EIS and \EISm remain reliable despite these large $\Delta_V,\Delta_W$ values: \EIS achieves less than $5\%$ mean error on \pubmed already with about 500~stored edges, while \EISm further reduces variance and delivers the best per-run accuracy on most datasets (with \dblp as the main exception, where subdividing the space budget degrades performance).

In \cref{sec:appendix_experiments}, we present additional experiments and plots. 
In particular, we present a detailed runtime breakdown and plots showing how the algorithms' relative error and variance behave as a function of \storedEdges $k$.

\section{Conclusion}
We studied four-cycle counting in arbitrary order graph streams. Our main contribution is a two-pass streaming algorithm, \EIS, that uses space~$\bigO(\kappa m/\sqrt{T})$, matching a theoretical space lower bound of $\Omega(m/\sqrt{T})$ for $\kappa=\bigO(1)$.
This algorithm performs very well in practice, using substantially less space to obtain small average errors than state-of-the-art baselines. We also presented a version with reduced variance, \EISm, which exhibits smaller variance than the algorithms we compare to. Further, we presented a one-pass algorithm, \NIS, which works well on unipartite social networks.

All of our algorithms are based on sampling a small induced subgraph and estimating the number of four-cycles subsequently. We believe that an intriguing question for further study is to develop more algorithms that follow this paradigm. As was shown in the past by \citet{ahmed2013NetworkSampling}, such induced subgraphs provide a lot of information about the overall graph topology. Thus, it is interesting to study whether such samples can also be used to estimate triangle counts, or clique counts, or other important graph properties.
If this is successful, by sampling a single induced subgraph we can solve multiple problems with little computational and memory overhead.

\section*{Acknowledgement}
    This research has been funded by the Vienna Science and Technology Fund (WWTF) [Grant ID: 10.47379/VRG23013]. PP is supported in part by NSFC Grant 62272431 and Quantum Science and Technology - National Science and Technology Major Project (Grant No. 2021ZD0302901).

\bibliographystyle{plainnat}
\bibliography{refs}

@inproceedings{mcgregorTriangleFourCycle2020,
  author       = {Andrew McGregor and Sofya Vorotnikova},
  title        = {Triangle and Four Cycle Counting in the Data Stream Model},
  booktitle    = {{PODS}},
  pages        = {445--456},
  publisher    = {{ACM}},
  year         = {2020}
}

@inproceedings{mcgregorBetterAlgorithmsCounting2016,
  author       = {Andrew McGregor and
                  Sofya Vorotnikova and
                  Hoa T. Vu},
  title        = {Better Algorithms for Counting Triangles in Data Streams},
  booktitle    = {{PODS}},
  pages        = {401--411},
  publisher    = {{ACM}},
  year         = {2016}
}

@inproceedings{papadiasCountingButterfliesFully2024,
  author       = {Serafeim Papadias and
                  Zoi Kaoudi and
                  Varun Pandey and
                  Jorge{-}Arnulfo Quian{\'{e}}{-}Ruiz and
                  Volker Markl},
  title        = {Counting Butterflies in Fully Dynamic Bipartite Graph Streams},
  booktitle    = {{ICDE}},
  pages        = {2917--2930},
  publisher    = {{IEEE}},
  year         = {2024}
}

@inproceedings{sunFABLEApproximateButterfly2024,
  author       = {Guozhang Sun and
                  Yuhai Zhao and
                  Yuan Li},
  title        = {{FABLE:} Approximate Butterfly Counting in Bipartite Graph Stream
                  with Duplicate Edges},
  booktitle    = {{CIKM}},
  pages        = {2158--2167},
  publisher    = {{ACM}},
  year         = {2024}
}

@article{mengCountingButterfliesStreaming2024,
  author       = {Lingkai Meng and
                  Long Yuan and
                  Xuemin Lin and
                  Chengjie Li and
                  Kai Wang and
                  Wenjie Zhang},
  title        = {Counting Butterflies over Streaming Bipartite Graphs with Duplicate
                  Edges},
  journal      = {CoRR},
  volume       = {abs/2412.11488},
  year         = {2024}
}

@article{liApproximatelyCountingButterflies2022,
  author       = {Rundong Li and
                  Pinghui Wang and
                  Peng Jia and
                  Xiangliang Zhang and
                  Junzhou Zhao and
                  Jing Tao and
                  Ye Yuan and
                  Xiaohong Guan},
  title        = {Approximately Counting Butterflies in Large Bipartite Graph Streams},
  journal      = {{IEEE} Trans. Knowl. Data Eng.},
  volume       = {34},
  number       = {12},
  pages        = {5621--5635},
  year         = {2022}
}

@inproceedings{sanei-mehriFLEETButterflyEstimation2019,
  author       = {Seyed{-}Vahid Sanei{-}Mehri and
                  Yu Zhang and
                  Ahmet Erdem Sariy{\"{u}}ce and
                  Srikanta Tirthapura},
  title        = {{FLEET:} Butterfly Estimation from a Bipartite Graph Stream},
  booktitle    = {{CIKM}},
  pages        = {1201--1210},
  publisher    = {{ACM}},
  year         = {2019}
}

@article{vorotnikovaImproved3passAlgorithm2020,
  author       = {Sofya Vorotnikova},
  title        = {Improved 3-pass Algorithm for Counting 4-cycles in Arbitrary Order
                  Streaming},
  journal      = {CoRR},
  volume       = {abs/2007.13466},
  year         = {2020}
}

@inproceedings{bordinoMiningLargeNetworks2008,
  author       = {Ilaria Bordino and
                  Debora Donato and
                  Aristides Gionis and
                  Stefano Leonardi},
  title        = {Mining Large Networks with Subgraph Counting},
  booktitle    = {{ICDM}},
  pages        = {737--742},
  publisher    = {{IEEE} Computer Society},
  year         = {2008}
}

@inproceedings{assadiSimpleSublinearTimeAlgorithm2018,
  author       = {Sepehr Assadi and
                  Michael Kapralov and
                  Sanjeev Khanna},
  title        = {A Simple Sublinear-Time Algorithm for Counting Arbitrary Subgraphs
                  via Edge Sampling},
  booktitle    = {{ITCS}},
  series       = {LIPIcs},
  volume       = {124},
  pages        = {6:1--6:20},
  publisher    = {Schloss Dagstuhl - Leibniz-Zentrum f{\"{u}}r Informatik},
  year         = {2019}
}

@inproceedings{fichtenbergerApproximatelyCountingSubgraphs2022,
  author       = {Hendrik Fichtenberger and
                  Pan Peng},
  title        = {Approximately Counting Subgraphs in Data Streams},
  booktitle    = {{PODS}},
  pages        = {413--425},
  publisher    = {{ACM}},
  year         = {2022}
}

@inproceedings{beraTighterSpaceBounds2017,
  author       = {Suman K. Bera and
                  Amit Chakrabarti},
  title        = {Towards Tighter Space Bounds for Counting Triangles and Other Substructures
                  in Graph Streams},
  booktitle    = {{STACS}},
  series       = {LIPIcs},
  volume       = {66},
  pages        = {11:1--11:14},
  publisher    = {Schloss Dagstuhl - Leibniz-Zentrum f{\"{u}}r Informatik},
  year         = {2017}
}

@inproceedings{beraHowDegeneracyHelps2020,
  author       = {Suman K. Bera and
                  C. Seshadhri},
  title        = {How the Degeneracy Helps for Triangle Counting in Graph Streams},
  booktitle    = {{PODS}},
  pages        = {457--467},
  publisher    = {{ACM}},
  year         = {2020}
}

@article{sheshboloukiSGrappButterflyApproximation2022,
  author       = {Aida Sheshbolouki and
                  M. Tamer {\"{O}}zsu},
  title        = {sGrapp: Butterfly Approximation in Streaming Graphs},
  journal      = {{ACM} Trans. Knowl. Discov. Data},
  volume       = {16},
  number       = {4},
  pages        = {76:1--76:43},
  year         = {2022}
}

@article{lalUnderstandingMoneyTrails2021,
  author       = {Banwari Lal and
                  Rachit Agarwal and
                  Sandeep Kumar Shukla},
  title        = {Understanding Money Trails of Suspicious Activities in a cryptocurrency-based
                  Blockchain},
  journal      = {CoRR},
  volume       = {abs/2108.11818},
  year         = {2021}
}

@article{oskarsdottirSocialNetworkAnalytics2022,
  author       = {Mar{\'{\i}}a {\'{O}}skarsd{\'{o}}ttir and
                  Waqas Ahmed and
                  Katrien Antonio and
                  Bart Baesens and
                  R{\'{e}}mi Dendievel and
                  Tom Donas and
                  Tom Reynkens},
  title        = {Social network analytics for supervised fraud detection in insurance},
  journal      = {CoRR},
  volume       = {abs/2009.08313},
  year         = {2020}
}

@article{chibaArboricitySubgraphListing1985,
  author       = {Norishige Chiba and
                  Takao Nishizeki},
  title        = {Arboricity and Subgraph Listing Algorithms},
  journal      = {{SIAM} J. Comput.},
  volume       = {14},
  number       = {1},
  pages        = {210--223},
  year         = {1985}
}

@inproceedings{schelter2016tracking,
  author       = {Sebastian Schelter and
                  J{\'{e}}r{\^{o}}me Kunegis},
  title        = {Tracking the Trackers: {A} Large-Scale Analysis of Embedded Web Trackers},
  booktitle    = {{ICWSM}},
  pages        = {679--682},
  publisher    = {{AAAI} Press},
  year         = {2016}
}

@inproceedings{kunegisKONECTKoblenzNetwork2013,
  author       = {J{\'{e}}r{\^{o}}me Kunegis},
  title        = {{KONECT:} the {Koblenz} network collection},
  booktitle    = {{WWW} (Companion Volume)},
  pages        = {1343--1350},
  publisher    = {International World Wide Web Conferences Steering Committee / {ACM}},
  year         = {2013}
}

@inproceedings{jayaramOptimalAlgorithmTriangle2021,
  author       = {Rajesh Jayaram and
                  John Kallaugher},
  title        = {An Optimal Algorithm for Triangle Counting in the Stream},
  booktitle    = {{APPROX-RANDOM}},
  series       = {LIPIcs},
  volume       = {207},
  pages        = {11:1--11:11},
  publisher    = {Schloss Dagstuhl - Leibniz-Zentrum f{\"{u}}r Informatik},
  year         = {2021}
}

@article{carter1979universal,
  author       = {Larry Carter and
                  Mark N. Wegman},
  title        = {Universal Classes of Hash Functions},
  journal      = {J. Comput. Syst. Sci.},
  volume       = {18},
  number       = {2},
  pages        = {143--154},
  year         = {1979}
}

@article{ahmed2013NetworkSampling,
  author       = {Nesreen K. Ahmed and
                  Jennifer Neville and
                  Ramana Rao Kompella},
  title        = {Network Sampling: From Static to Streaming Graphs},
  journal      = {{ACM} Trans. Knowl. Discov. Data},
  volume       = {8},
  number       = {2},
  pages        = {7:1--7:56},
  year         = {2013}
}

@article{charbey2019stars,
  author       = {Rapha{\"{e}}l Charbey and
                  Christophe Prieur},
  title        = {Stars, holes, or paths across your Facebook friends: {A} graphlet-based
                  characterization of many networks},
  journal      = {Netw. Sci.},
  volume       = {7},
  number       = {4},
  pages        = {476--497},
  year         = {2019}
}

@inproceedings{rotabi2017detectingStrong,
  author       = {Rahmtin Rotabi and
                  Krishna Kamath and
                  Jon M. Kleinberg and
                  Aneesh Sharma},
  title        = {Detecting Strong Ties Using Network Motifs},
  booktitle    = {{WWW} (Companion Volume)},
  pages        = {983--992},
  publisher    = {{ACM}},
  year         = {2017}
}

@article{sporns2004motifs,
  title={Motifs in brain networks},
  author={Sporns, Olaf and K{\"o}tter, Rolf},
  journal={PLoS Biology},
  volume={2},
  number={11},
  pages={e369},
  year={2004},
  publisher={Public Library of Science San Francisco, USA}
}

@article{shin2018Patterns,
  author       = {Kijung Shin and
                  Tina Eliassi{-}Rad and
                  Christos Faloutsos},
  title        = {Patterns and anomalies in k-cores of real-world graphs with applications},
  journal      = {Knowl. Inf. Syst.},
  volume       = {54},
  number       = {3},
  pages        = {677--710},
  year         = {2018}
}

@inproceedings{kallaugherMPV19,
  author       = {John Kallaugher and
                  Andrew McGregor and
                  Eric Price and
                  Sofya Vorotnikova},
  title        = {The Complexity of Counting Cycles in the Adjacency List Streaming
                  Model},
  booktitle    = {{PODS}},
  pages        = {119--133},
  publisher    = {{ACM}},
  year         = {2019}
}

@article{lind2005cycles,
  title={Cycles and clustering in bipartite networks},
  author={Lind, Pedro G. and Gonzalez, Marta C. and Herrmann, Hans J.},
  journal={Phys. Rev. E Stat. Nonlin. Soft Matter Phys.},
  volume={72},
  number={5},
  pages={056127},
  numpages = {9},
  year={2005},
  publisher={APS}
}

@inproceedings{abuodaMA19LinkPrediction,
  author       = {Ghadeer Abuoda and
                  Gianmarco De Francisci Morales and
                  Ashraf Aboulnaga},
  title        = {Link Prediction via Higher-Order Motif Features},
  booktitle    = {{ECML/PKDD} {(1)}},
  series       = {Lecture Notes in Computer Science},
  volume       = {11906},
  pages        = {412--429},
  publisher    = {Springer},
  year         = {2019}
}

@article{eppstein2013listing,
  author       = {David Eppstein and
                  Maarten L{\"{o}}ffler and
                  Darren Strash},
  title        = {Listing All Maximal Cliques in Large Sparse Real-World Graphs},
  journal      = {{ACM} J. Exp. Algorithmics},
  volume       = {18},
  year         = {2013},
}

@article{alon1999space,
  author       = {Noga Alon and
                  Yossi Matias and
                  Mario Szegedy},
  title        = {The Space Complexity of Approximating the Frequency Moments},
  journal      = {J. Comput. Syst. Sci.},
  volume       = {58},
  number       = {1},
  pages        = {137--147},
  year         = {1999}
}

@article{shridhar2025network,
  author    = {Shridhar, Shivkumar Vishnempet and Lee, Selena T. and Charette, Yanick and Iosifidis, George and Christakis, Nicholas A.},
  title     = {Network-cycle motif participation is associated with individual and collective wealth in {Honduran} villages},
  journal   = {Scientific Reports},
  volume    = {15},
  number    = {1},
  pages     = {27680},
  year      = {2025}
}

@inproceedings{stefani2016triest,
  author       = {Lorenzo De Stefani and
                  Alessandro Epasto and
                  Matteo Riondato and
                  Eli Upfal},
  title        = {TRI{\`{E}}ST: Counting Local and Global Triangles in Fully-Dynamic
                  Streams with Fixed Memory Size},
  booktitle    = {{KDD}},
  pages        = {825--834},
  publisher    = {{ACM}},
  year         = {2016}
}

@inproceedings{lim2015mascot,
  author       = {Yongsub Lim and
                  U Kang},
  title        = {{MASCOT:} Memory-efficient and Accurate Sampling for Counting Local
                  Triangles in Graph Streams},
  booktitle    = {{KDD}},
  pages        = {685--694},
  publisher    = {{ACM}},
  year         = {2015}
}

@article{vitter85random,
  author       = {Jeffrey Scott Vitter},
  title        = {Random Sampling with a Reservoir},
  journal      = {{ACM} Trans. Math. Softw.},
  volume       = {11},
  number       = {1},
  pages        = {37--57},
  year         = {1985}
}

@inproceedings{braverman2013hard,
  author       = {Vladimir Braverman and
                  Rafail Ostrovsky and
                  Dan Vilenchik},
  title        = {How Hard Is Counting Triangles in the Streaming Model?},
  booktitle    = {{ICALP} {(1)}},
  series       = {Lecture Notes in Computer Science},
  volume       = {7965},
  pages        = {244--254},
  publisher    = {Springer},
  year         = {2013}
}

@inproceedings{nr,
      title = {The Network Data Repository with Interactive Graph Analytics and Visualization},
      author={Ryan A. Rossi and Nesreen K. Ahmed},
      booktitle = {AAAI},
      url={http://networkrepository.com},
      year={2015}
}

@inproceedings{luederssen2026near,
    author = {Sebastian L{\"{u}}derssen and Stefan Neumann and Pan Peng},
    title = {Near-Optimal Four-Cycle Counting in Graph Streams},
    booktitle    = {{SODA}},
    publisher    = {{SIAM}},
    year         = {2026},
    pages        = {4285--4326}
}

@misc{code,
    author = {Sebastian L{\"{u}}derssen and Stefan Neumann and Pan Peng},
  title = {Code for the paper ``{Four-Cycle Counting in Low-Degeneracy Graph Streams}''},
  year = {2026},
  publisher = {Zenodo},
  journal = {Zenodo repository},
  howpublished = {\url{https://doi.org/10.5281/zenodo.18089153}},
}

\appendix
\pagebreak

\section{Omitted content from \texorpdfstring{\cref{sec:algorithms}}{Section 3}}
\label{sec:appendix_algorithms}

We present missing content from \Cref{sec:algorithms}.

\cref{algo:nis} presents the pseudocode of the \NIS algorithm presented in \cref{sec:algorithms}. 

\begin{algorithm}[ht]
    \caption{\NISlong \NIS}
    \label{algo:nis}
    \DontPrintSemicolon
    \textbf{Input:} Graph $G$, sampling probability $p$\;
    Initialize an 8-wise-independent hash function $h:V\rightarrow[0,1]$\;
    $E_h=\emptyset$\;
    \textbf{Pass 1:} \ForEach{$e=(v,w)\in E$}{
        \If{$h(v)< p$ and $h(w)< p$}{
            Add $e$ to $E_h$\;
        }
    }
    \textbf{After Pass 1:}
    Let $X$ be the number of four-cycles in $G[V_h]$\;
    \Return{$\hat{T}=X/p^4$}
\end{algorithm}

\paragraph{\NIS with fixed size budget.}
We now give a detailed description of the fixed-size budget variant of \NIS.
Similarly to the fixed-size variant of \EIS, we adapt \NIS by dynamically deleting from the sample.
However, instead of using reservoir sampling, we manually vary the sampling probability $p$ in the algorithm.

We start by running \NIS with $p=1$. and we collect edges until $G[V_h]$ contains
$k$~edges.  Now whenever the limit of $k$ induced edges in~$G[V_h]$ is exceeded,
we decrease~$p$ and with it adjust the sampled node set
$V_h=\{v\in V \colon h(v)<p\}$ and induced edge set 
$E_h=\{(v,w)\in E \colon h(v)<p \land h(w)<p\}$. Note that decreasing~$p$
can never increase the size of $E_h$.
To find a lower value for $p$ such that $G[V_h]$ contains less than $k$~edges,
we set $p$~to the highest hash value~$h(v)$
of the nodes $v\in V_h$. We then remove all nodes~$v$ from $V_h$ with $h(v)\geq
p$ and also delete all of their incident edges from $G[V_h]$ as well.

To implement this procedure efficiently, each edge $e=(u,v)$ obtains a score
$\theta(e)=\max\{h(u),h(v)\}$. We maintain the scores in a binary search tree
and whenever we need to decrease~$p$, we set $p$ to the largest score and remove
all edges with score~$\theta(e)\geq p$. Note that we do not store the score of
all seen edges, but instead maintain a list of the at most $k$~edges currently
stored in $G[V_h]$.  Upon arrival of $e$ in the stream, we compute its score
$\theta(e)$, and insert it in the search tree if it is also added to $G[V_h]$.
Note that the size of the tree is always the same as that of $G[V_h]$, and thus never
contains more than $k$~elements.

Once the stream pass is finished, we again have to decide which value of~$p$ to use
for the rescaling of our estimate.  As before, we let $X$ denote the number of
four-cycles $X$ in $G[V_h]$. Now consider the value of~$p$ at the end of the
stream pass.  While defining $\hat{T}=X/p^4$ is the most natural choice, it does
not give the best results in practice. Instead, we use the plug-in estimate 
$$\hat{T}=\frac{m^2}{|E_h|^2}X.$$
This choice is motivated by the fixed-probability version, where
$\mathbb{E}[|E_h|]=mp^2$ and the unbiased scaling is $1/p^4$.
In the adaptive fixed-budget implementation, this plug-in scaling is heuristic,
but it gives better empirical performance than using the final threshold $p$
directly.

\section{Omitted content from \texorpdfstring{\cref{sec:theory}}{Section 4}}
\label{sec:appendix_proofs}

Next, we present missing content from \Cref{sec:theory}.

\subsection{Pseudocode of the theoretical version of \EIS}
\label{sec:pseudocode-EIS-theory}

We state the pseudocode of the theoretical version of \EIS in \Cref{algo:eis_oracle}.

\begin{algorithm}[ht]
    \caption{\EIS with heaviness oracle}
    \label{algo:eis_oracle}
    \DontPrintSemicolon
    \textbf{Pass 1:} 
        Sample edge sets~$S$ and~$R$ by sampling each edge twice u.a.r.\ with probability~$p=320\log(n)/(\epsilon^3\sqrt{T})$\;
    Let $V_S$ be the nodes incident to at least one edge in~$S$\;
    Let $V_R$ be the nodes incident to at least one edge in~$R$\;
    \textbf{Pass 2:} Collect edges in $G[V_S\cup V_R]$\;
    Compute $\operatorname{oracle}(e) \in \{\textsf{heavy}, \textsf{light}\}$ for all $e\in G[V_S]$\;
    Let $A_0,A_1=0$\;
    \ForEach{pair of oracle-light and vertex-disjoint edges $e,f\in S$}{
        Increase $A_0$ by the number of four-cycles induced by $e$ and $f$ without oracle-heavy edges\;
        Increase $A_1$ by the number of four-cycles induced by $e$ and $f$ with exactly one oracle-heavy edge\;
    }
    \Return{$A_0/(2p^2)+A_1/p^2$}
\end{algorithm}

\subsection{Proof of \texorpdfstring{\cref{lemma:oracle_accuracy}}{Lemma 2}}
    Consider an edge $e\in G[V_S]$. Let $(X_f)_{f\in E}$ be the random variables counting the number of four-cycles induced by $e$ and $f$, if $f\in R$, and $0$, otherwise.
    The oracle classifies an edge~$e$ by computing $q(e)=\sum_{f\in E} X_f$ and returning heavy or light based on \Cref{eq:oracle}.
    
Observe that the variables $(X_f)_{f\in E}$ are mutually independent, since they
depend on independent sampling decisions for the edges in $R$. Moreover,
$0\le X_f\le 2$, and 
\begin{align*}
        \IE[q(e)]=\sum_{\substack{f\in E \\ \text{$e,f$ induce a four-cycle}}} \EXP{X_f} = pt(e).
    \end{align*}

We apply a standard multiplicative concentration bound for independent bounded
random variables after scaling by a factor of $2$. 

If $t(e)>4\eta\sqrt{T}$, then
    \begin{equation*}
        \Prob{q(e)< \left(1-\frac{1}{2}\right)4p\eta\sqrt{T}}
		\leq \exp\left(-\frac{1}{6\cdot 4}4p\eta\sqrt{T}\right)
		\leq\frac{1}{n^3},
    \end{equation*}
    where we used $p=320\log(n)/(\eps^3\sqrt{T})$.
    
    If $t(e)<\eta\sqrt{T}$, then
    \begin{equation*}
        \Prob{q(e)>(1+1)p\eta\sqrt{T}}
		\leq \exp\left(-\frac{1}{6}p\eta\sqrt{T}\right)
	\leq\frac{1}{n^3}.
    \end{equation*}
    A union bound over all edges yields the statement.

\subsection{Proof of \texorpdfstring{\cref{lemma:eis_variance}}{Lemma 3}}
    Let $i\in\{0,1\}$. By Chebyshev's inequality, we obtain that
    \begin{equation*}
        \Prob{|\hat{T_i}-T_i|\geq \epsilon T}\leq \frac{\IV[\hat{T_i}]}{\epsilon^2T^2}.
    \end{equation*}

    Now we can bound the variance as follows:
    \begin{align}
        \Var{\hat{T_i}}&\leq \Var{\frac{A_i}{p^2}}=\frac{\IV[A_i]}{p^4} \leq \frac{\EXP{A_i^2}}{p^4}\nonumber \\
        &=\frac{1}{p^4}\left( \sum_{q\in D_i} \EXP{X_q^2} + \sum_{\substack{q, r\in D_i\\ |q\cap r|=1}}\EXP{X_q X_r} \right) \nonumber \\
        &\leq\frac{1}{p^4}\left( 4|D_i|p^2 + \sum_{q\in D_i} \sum_{\substack{r\in D_i\\ |q\cap r|=1}} 4p^3 \right) \nonumber \\
        &\leq \frac{1}{p^4}\left( 4|D_i|p^2 + |D_i| (2\cdot4\eta \sqrt{T})4p^3  \right) \label{eq:c4_1}\\
        &\leq \frac{4|D_i|}{p^2} + \frac{32\eta |D_i|\sqrt{T}}{p} \nonumber\\
        &\leq \frac{1}{10}\epsilon^2T^2 \label{eq:c4_2},
    \end{align}
    where we used that each pair of oracle-light edges can be contained in at most $2\cdot4\eta \sqrt{T}$ four-cycles by \Cref{lemma:oracle_accuracy}, which implies \eqref{eq:c4_1}. The final inequality~\eqref{eq:c4_2} and the statement follow from $|D_i|\leq 2T$ and setting $p=320\log(n)/(\epsilon^3\sqrt{T})$.

\subsection{Proof of \texorpdfstring{\cref{thm:NIS}}{Theorem 2}}
\label{sec:proof-NIS}

To prove the theorem, we set the sampling probability as  $$p=\frac{32}{(\epsilon^2T^{1/4})}.$$

To bound the space usage of our algorithm,
recall that storing the 8-wise-independent hash function can be done in $O(\log n)$ bits of space \cite{carter1979universal}.
Next, note that each edge is sampled with probability~$p^2$ as both of its endpoints are part of $V_h$ with probability $p$ independently. Here, we use the fact that our chosen hash function is 2-wise-independent. 
Thus \NIS stores $s = O(mp^2)=O(m/(\epsilon^4\sqrt{T}))$ edges in expectation. 
Thus, by Markov's inequality \NIS stores at most $8s=O(m/\eps^4 \sqrt{T})$ edges with probability at least~$\frac{7}{8}$.

Next, we show that our estimator is unbiased, i.e., that $\EXP{\hat{T}}=T$.
Recall that \NIS computes an estimate $\hat{T} = X/p^4$ on the number of four-cycles $T$ in $G$ and, as briefly discussed in \Cref{sec:algorithms},
any four-cycle is part of $G[V_h]$ if and only if all of its four vertices are sampled. For this, all four nodes of the four-cycle need to be part of $V_h=\{v\in V: h(v)\leq p\}$, which happens with probability $p^4$, as we are using a 4-wise independent hash function $h$.
Thus, any four-cycle is counted with probability $p^4$, and thus
\begin{align*}
    \EXP{\hat{T}}&= \frac{1}{p^4}\sum_{A\in\squares} p^4 = T.
\end{align*}

Next, we bound the variance of \NIS.
\begin{lemma}
The variance of the estimate $\hat{T}$ of \NIS is bounded by
\begin{align}
\label{eq:nis_variance}
    \Var{\hat{T}}\leq\frac{T}{p^4}+\frac{4T\Delta_V}{p}+\frac{4T\Delta_E}{p^2}+\frac{4T\Delta_W^2}{p^2}+\frac{4T\Delta_W}{p^3}.
\end{align}
\end{lemma}
\begin{proof}
    We analyze the variance by considering pairs of intersecting four-cycles.
	Define the binary random variable $X_A$ for any four-cycle $A$ to be~1, 
	if all nodes of $A$ were sampled in $V_h$, and 0, otherwise. Then $X=\sum_{A\in\squares} X_A$. For simplicity, we write $B\cap A:=V(A)\cap V(B)$ to denote the shared nodes of two four-cycles $A,B$.
    Now, as we are using an 8-wise-independent hash function to determine our node sample, any pair of variables $X_A$ and $X_B$ where the four-cycles $A$ and $B$ do not share vertices is independent. Using this property, we can bound the variance as
    \begin{align*}
        \Var{\hat{T}}&\leq \frac{1}{p^8} \Var{X} \leq \frac{1}{p^8} \EXP{X^2}\\
        &\leq \frac{1}{p^8}\left(
				\sum_{A\in\squares}\sum_{\substack{B\in\squares\\|B \cap A|>0}} \EXP{X_AX_B} \right)\\
        &\leq \frac{1}{p^8}\left(
				\sum_{A\in\squares}\sum_{i=1}^4\sum_{\substack{B\in\squares\\|B \cap A|=i}} \EXP{X_AX_B} \right)\\
        &\leq \frac{1}{p^8}\left( Tp^4 \left(4\Delta_V p^3 + 4(\Delta_E+\Delta_W^2) p^2 + \Delta_W p\right) \right).
        \qedhere
    \end{align*}
\end{proof}

\begin{proof}[Proof of \Cref{thm:NIS}]
    Using Chebyshev's inequality, we get
    $$\Prob{|\hat{T}-T|\geq \epsilon T}\leq \frac{\Var{\hat{T}}}{\epsilon^2T^2}.$$ 
    Thus, we obtain a $(1+\varepsilon)$-approximation with probability $3/4$ if $\Var{\hat{T}}\leq \epsilon^2T^2/4$.
    The latter follows if 
    the input graph $G$ fulfills $\Delta_V\leq T^{3/4}$, $\Delta_E\leq \sqrt{T}$, and $\Delta_W\leq T^{1/4}$ due to \Cref{eq:nis_variance} and  by setting $p=32/(\epsilon^2 T^{1/4})$.
    Taking a union bound with the error probability for the space usage, we obtain the theorem.
\end{proof}

\section{Further Experiment Details}
\label{sec:appendix_experiments}
We provide further details about our experiments.

\subsection{Dataset Description}
\label{sec:datasets}

In our experiments, we use the two bipartite networks \orkut and \livejournalgroups, which model group membership of users in social networks. Nodes are users and groups, resp., and edges link users to their groups. These networks behave fundamentally differently from friendship networks, as single groups contain more than one million users. Thus, two similar groups with a large intersection form many heavy wedges. On \livejournalgroups there are just two groups, which form extremely heavy wedges and already generate $6\%$ of the four-cycles in the network.

Next, the bipartite datasets with the highest average degrees (each of them with $\bar{d}\geq 100$) are \movielens, \reuters and \pubmed. The \movielens dataset is a well-known bipartite network modeling movie ratings and is particularly dense ($\bar{d}\geq 240$  and $\kappa=476$); additionally, it has a large number of four-cycles, given its 10~million edges.
\reuters and \pubmed model word-text-inclusion.
All three of these networks have large values of $\Delta_V$ and $\Delta_W$, suggesting that \NIS should perform badly on them, but their $\Delta_E$-values are still moderate, suggesting that \EIS and \EISm should perform well.

\begin{figure*}[t]
  \centering
  \includegraphics[width=\textwidth]{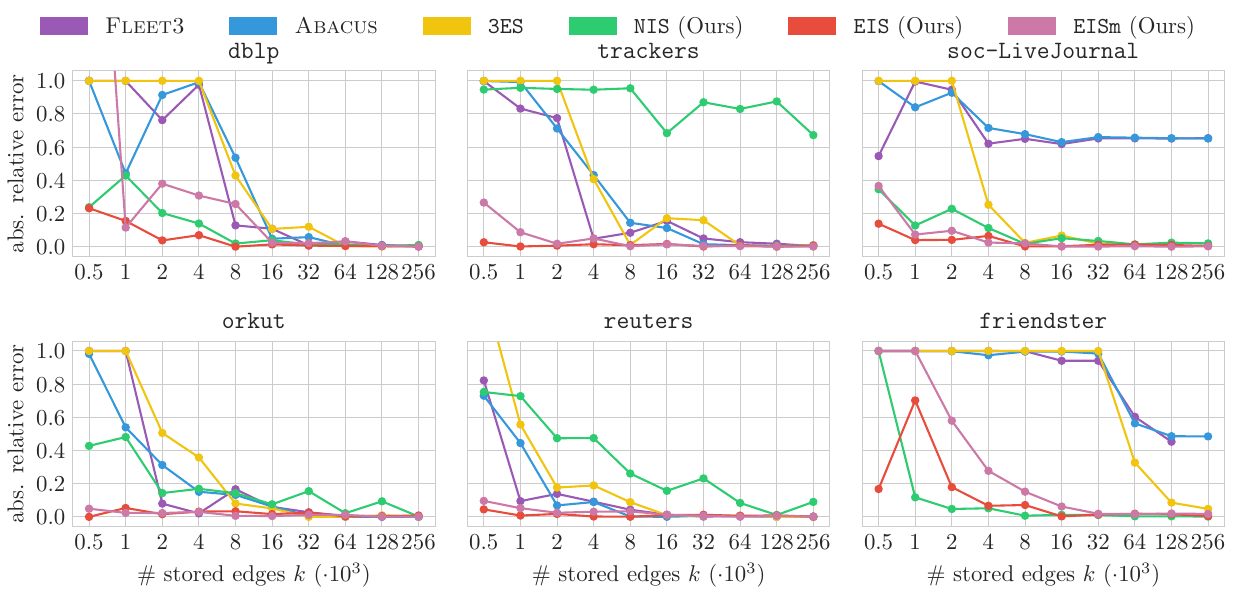}
  \caption{Relative error of the algorithms as a function of \storedEdges on selected instances (averaged over 100 runs). \EIS obtains the best results and consistently achieves very low errors, even when only storing 2k~edges or less.}
  \label{fig:algocompare_mean}
\end{figure*}

\begin{figure*}[t]
  \centering
  \includegraphics[width=\textwidth]{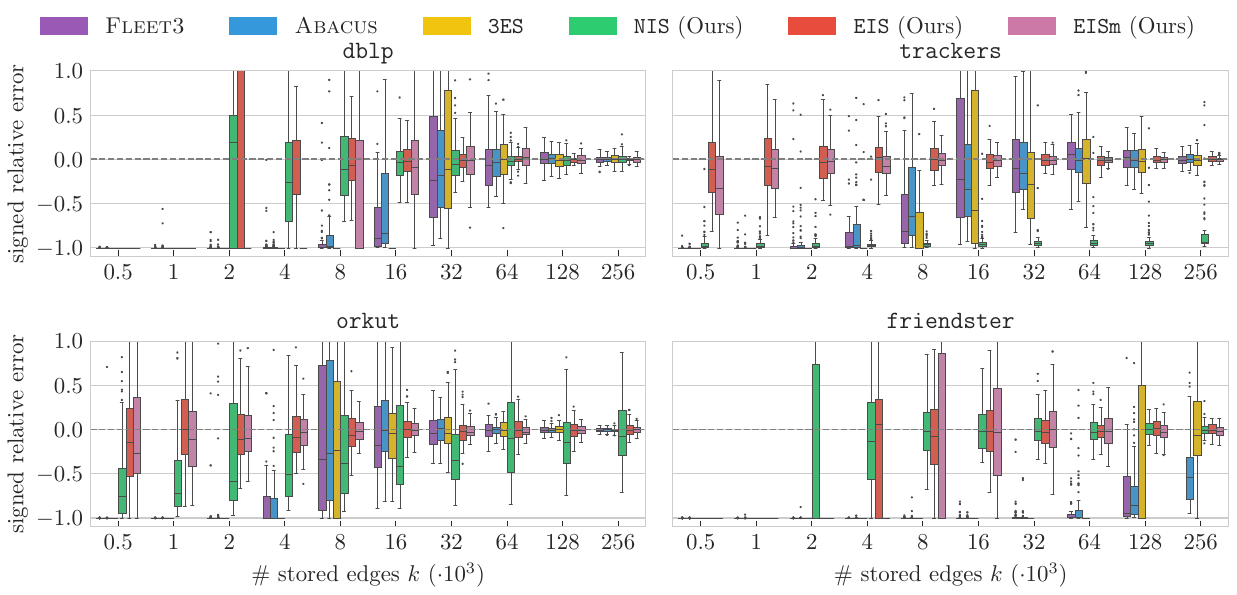}
  \caption{Variance of the relative error of the algorithms on selected instances. Based on 100 estimates, the box spans from the 25th to the 75th percentile, with the line inside indicating the median. Dots represent outliers, and vertical lines show the remaining values in the first and fourth quartiles.}
  \label{fig:algocompare_variance}
\end{figure*}

Additionally, we consider two very sparse datasets ($\bar{d} \leq 10$), \dblp and \trackers. Even though both have low average degree, they are quite different. The \dblp authorship network contains only very few four-cycles (roughly $T\approx 30M$), and its four-cycles are well-distributed (as evidenced by small values for $\Delta_V$, $\Delta_E$ and $\Delta_W$).
In contrast, the \trackers dataset links website and website-trackers of a 2012 web crawl, and, while having a low average degree, the four-cycle distribution is highly skewed towards a few nodes. Indeed, 25\% of all websites in the dataset contain a \texttt{googleapis.com} tracker and the corresponding vertex is contained in more than half of all four-cycles in the instance \cite{schelter2016tracking}.  This leads to very high values of $\Delta_V$ and $\Delta_W$ (both are $\geq 1000$, by far the largest values across all of our datasets), but nonetheless $\Delta_E$ is small.

We also consider three unipartite social networks \flickr, \socLiveJournal, and \friendster, which model friendship or follower connections between users. Naturally, such networks have low density and even distributions of four-cycles, which is confirmed by the low values for the heaviness values $\Delta_V$, $\Delta_E$, $\Delta_W$. The \friendster network is our largest instance with 1.8~billion unique edges.

Finally, we consider three datasets with high degeneracy. The genetic networks \biohuman and \bnhuman and the collaboration network \hollywood have a degeneracy up to $\kappa=2208$, while all previously discussed datasets had a degeneracy of at most 600.

\subsection{Further average error analysis}
\label{sec:appendix_experiments_mean}

In \Cref{fig:algocompare_mean} we present the absolute relative error as a function of the \storedEdges~$k$ on six selected instances, including two of the unipartite networks (on the right). 
\EIS is consistently the best method for a given space usage, with \EISm the runner-up. Typically, the estimates of \EIS and \EISm converge much more quickly than for the competing methods, and outperform the two-pass baseline \multipassbaseline. 
As suggested by our theoretical analysis, the performance of \NIS depends heavily on the $\Delta_V$ and $\Delta_W$ values of the datasets.

\subsection{Further variance analysis}
\label{sec:appendix_experiments_variance}

In \Cref{fig:algocompare_variance} we report the variance of all algorithms for increasing \storedEdges for a representative subset of instances. 
We observe that on all instances (except the very sparse \dblp, see below), \EISm has a smaller variance than \EIS, as also suggested by \Cref{tab:variance_table} and by how we designed \EISm. This is enabled by the fact that \EIS obtains excellent mean errors even for very little space usage, and thus running \EIS with little space multiple times (as done in \EISm) will yield estimates of lower variance.

Furthermore, we observe that for very small sample sizes, all algorithms exhibit a large variance in performance. 
While on \dblp all algorithms converge to low variance estimators,  on \orkut the variance of \NIS remains high. 
Further, on \orkut \Abacus and \Fleet give accurate estimates more consistently than \EIS for large sample sizes, while \EIS performs better for a sample size of 16k.
We also observe that on \trackers, \NIS not only has high variance, but also fails to sample any four-cycle even for high sample rates, which matches the observation that most of the four-cycles are concentrated on a few vertices.
On \friendster, \Fleet and \Abacus do not find any four-cycles for sample sizes up to 64k, while both \NIS and \EIS can give adequate estimates.

A final important observation is that on the sparse \dblp network the regular \EIS clearly outperforms the variance-reduced \EISm. 
We explore an explanation for this phenomenon next.

\subsection{\EIS outperforms \EISm on very sparse graphs}
\label{sec:appendix_experiments_inducedFraction}

\begin{figure}[t]
    \centering
    \includegraphics[width=0.9\linewidth]{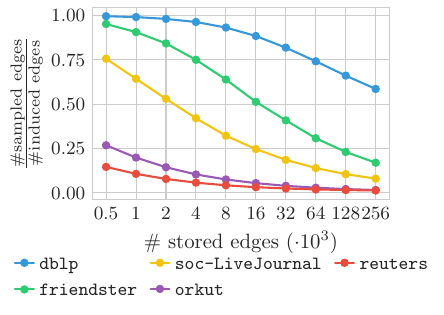}
    \caption{Fraction~$|S|/|G[V_S]|$ of sampled edges over induced edges when running \EIS on selected instances. For large graph sample sizes, \EIS spends a lot more space on storing $G[V_S]$ compared to $S$, making the estimator less powerful. }
    \label{fig:density}
\end{figure}

We now give an explanation for the following two observations from \Cref{tab:variance_table} and \Cref{fig:algocompare_variance}:
\begin{enumerate}
    \item \EIS outperforms \EISm on the \dblp network.
    \item On some instances \EIS exhibits a larger variance than the baselines, even for large sample sizes.
\end{enumerate}
In \Cref{fig:density} we plot the average ratio of sampled edges~$S$ over induced edges in $G[V_S]$ when running \EIS on different instances, i.e., we report $\frac{|S|}{|G[V_S]|}$. Note that as each sampled edge is also an induced edge, the ratio is at most 1. 
Observe that this ratio is meaningful since we only count four-cycles that are induced by two edges from~$S$. Thus, if the ratio is small, we expect that \EIS is ``wasting'' a lot of space on storing edges in $G[V_S]$ that do not contribute to obtaining better estimates.

In the figure, we see that the ratio $\frac{|S|}{|G[V_S]|}$ decreases monotonically on all instances, but behaves very differently across instances.
On \dblp more than 50\% of the edges in the induced subgraph are also in~$S$, even for $256.000$ stored edges, and for small \storedEdges a ratio of nearly 1 indicates that no induced edges are collected at all due to the low density of the network.
On \reuters, on the other hand, less than $1.5\%$ of subgraph edges are edges from~$S$ for $k=256.000$ and thus the vast majority of available space is used for storing induced edges.

We attribute the slow convergence of the variance of \EIS with increasing \storedEdges to the fact that the proportion of sampled edges in the graph sample decreases considerably with each increase in \storedEdges, implying diminishing returns in estimator accuracy.
This explains the effectiveness of \EISm on most networks and why it is outperformed by \EIS on \dblp.

\subsection{Sensitivity analysis of parameter \texorpdfstring{$s$}{s}}
\label{sec:appendix_experiments_sensitivity-s}

\begin{figure*}[t]
    \centering
    \includegraphics[width=\linewidth]{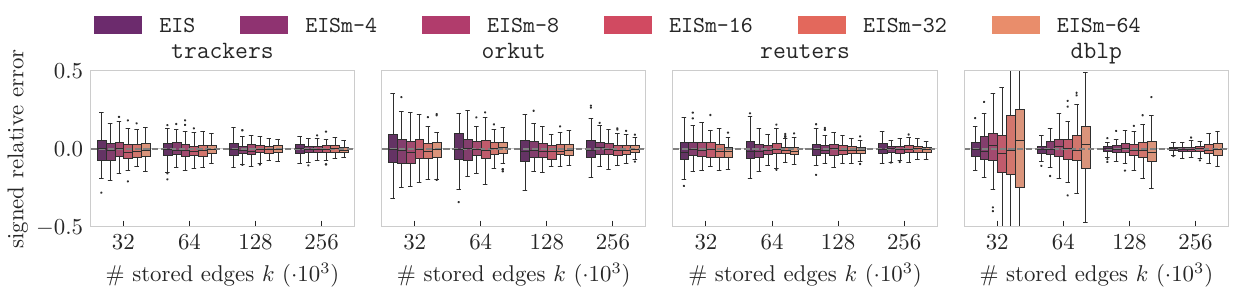}
    \caption{Sensitivity analysis of \EISm on the parameter $s$. Increasing $s$ reduces the variance on most instances, as expected. On the very sparse \dblp network increasing $s$ yields worse results, as explained in \cref{sec:appendix_experiments_inducedFraction}.}
    \label{fig:varying_s}
\end{figure*}

Recall that the \EISm algorithm divides the available space budget $k$ into $s$ parts and averages the estimates of $s$ parallel runs of $\EIS$, each using space $k/s$. 
In the experiments in \cref{sec:experiments}, we set $s=32$ when comparing \EISm to the other algorithms. We now justify this choice of $s$ by performing a sensitivity analysis of \EISm on the parameter $s$. 

We analyze the variance of \EISm for different values of $s$. 
We run \EISm with parameter $s\in\{4,8,16,32,64\}$ and compare the variance of the resulting estimates with \EIS.
\Cref{fig:varying_s} shows the results for a selection of instances. We plot the variance of all estimators for values $k\geq 32.000$.

On \trackers, \reuters and \orkut, the variance decreases with increasing $s$. Thus, as expected, \EISm gets more accurate for larger values of $s$, at the expense of an increase in runtime linear in $s$.
However, the gain in accuracy is diminishing as $s$ grows, leading to only slightly lower variance for $s=64$ compared to $s=32$.

Interestingly, on the \dblp network, increasing $s$ degrades the performance of the \EISm estimator. This observation matches our analysis from \cref{sec:appendix_experiments_inducedFraction}: \dblp is so sparse that even for large $k$, the number of induced edges collected by \EIS is small. Thus, using \EISm and subdividing the available size budget results in even fewer induced edges and less knowledge of the graph structure. This explains the increased variance when increasing $s$ on \dblp.

Overall, we conclude that as the runtime of \EISm scales linearly in $s$, choosing $s=32$ is a reasonable trade-off between runtime and variance improvements.

\subsection{Runtime Breakdown}

In \Cref{fig:runtime_breakdown}, we provide a runtime breakdown of our implementation of \EIS. 

\begin{figure*}[t]
    \centering
    \includegraphics[width=0.7\linewidth]{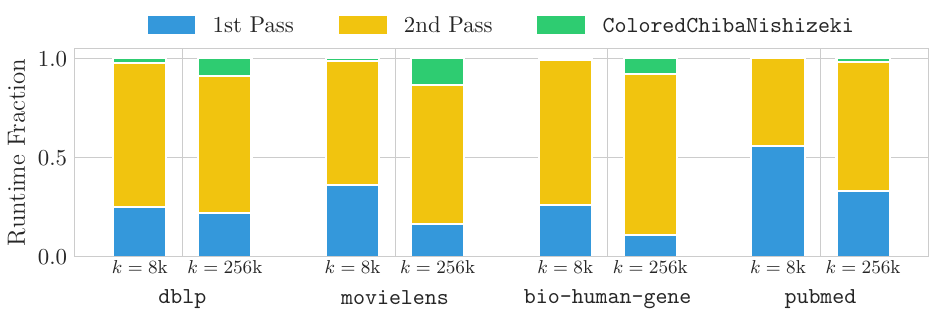}
    \caption{Runtime Breakdown for \EIS on selected instances for two values of \storedEdges.
    For each network, we show the relative time spent during the first pass, the second pass and the postprocessing using \texttt{ColoredChibaNishizeki}.}
    \label{fig:runtime_breakdown}
\end{figure*}

Across all instances, most time is spent during the second stream pass. The time spent for counting four-cycles in the sample using \texttt{ColoredChibaNishizeki} is negligible for small $k$ and remains low even for $k=256$k.
Additionally, we observe that as $k$~increases, the first pass takes a smaller fraction of the overall runtime. This is caused by an increased number of hashmap queries in the second pass and increased running time for \texttt{ColoredChibaNishizeki} due to the larger sample size.
On the \pubmed dataset, which is the largest dataset considered in the figure (containing 483M edges), the first pass requires a larger fraction of the running time since the first pass requires a random sample for each edge in the input; in contrast, the hashing overhead in the second pass only scales with the size of the sample, which is fixed, and therefore does not increase with the number of edges.

\end{document}